\documentclass[
notitlepage
,floatfix,
aps,
prl,
reprint,
superscriptaddress,
]{revtex4-2}
\usepackage[utf8]{inputenc}

\usepackage[caption=false]{subfig}
\usepackage{graphicx}
\usepackage{epstopdf}
\usepackage{array}
\usepackage{verbatim}
\usepackage{amsmath,amsfonts,amssymb,amscd,mathtools}
\usepackage{amsthm}
\usepackage{tabularx}
\usepackage{stmaryrd}
\usepackage{enumerate}
\usepackage{wasysym}
\usepackage{braket}
\usepackage{mathrsfs}
\usepackage{microtype}
\usepackage{hyperref}
\usepackage[dvipsnames]{xcolor}
\hypersetup{
    bookmarksnumbered=true, 
    unicode=false, 
    pdfstartview={FitH}, 
    pdftitle={}, 
    pdfauthor={}, 
    pdfsubject={}, 
    pdfcreator={}, 
    pdfproducer={}, 
    pdfkeywords={}, 
    pdfnewwindow=true, 
    colorlinks=true, 
    linkcolor=NavyBlue, 
    citecolor=NavyBlue, 
    filecolor=NavyBlue, 
    urlcolor=NavyBlue 
}
\usepackage{tikz}
\usepackage[lmargin=.7in,rmargin=.7in,tmargin=.7in,bmargin=1in]{geometry}
\usepackage{newtxtext,newtxmath}

\usepackage{algorithm}
\usepackage{algorithmicx}
\usepackage{algpseudocode}



\usepackage{booktabs}

\theoremstyle{plain}
\newtheorem{thm}{Theorem}

\newtheorem{pro}[thm]{Proposition}

\theoremstyle{definition}

\newcommand{\eq}[1]{(\hyperref[eq:#1]{\ref*{eq:#1}})}

\renewcommand{\sec}[1]{\hyperref[sec:#1]{Section~\ref*{sec:#1}}}
\newcommand{\thrm}[1]{\hyperref[thrm:#1]{Theorem~\ref*{thrm:#1}}}
\newcommand{\lemm}[1]{\hyperref[lemm:#1]{Lemma~\ref*{lemm:#1}}}
\newcommand{\prop}[1]{\hyperref[prop:#1]{Proposition~\ref*{prop:#1}}}
\newcommand{\corr}[1]{\hyperref[corr:#1]{Corollary~\ref*{corr:#1}}}
\newcommand{\fig}[1]{\hyperref[fig:#1]{~\ref*{fig:#1}}}
\newcommand{\deff}[1]{\hyperref[deff:#1]{~\ref*{deff:#1}}}

\newcommand{\mA}{\mathcal{A}}

\newcommand{\mE}{\mathcal{E}}
\newcommand{\mN}{\mathcal{N}}
\newcommand{\mU}{\mathcal{U}}
\newcommand{\mT}{\mathcal{T}}
\newcommand{\mD}{\mathcal{D}}

\newcommand{\mF}{\mathcal{F}}
\newcommand{\mL}{\mathcal{L}}

\newcommand{\mO}{\mathcal{O}}
\newcommand{\mB}{\mathcal{B}}

\newcommand{\mP}{\mathcal{P}}

\newcommand{\mX}{\mathcal{X}}
\newcommand{\mY}{\mathcal{Y}}
\newcommand{\mZ}{\mathcal{Z}}

\newcommand{\mbD}{\mathbb{D}}
\newcommand{\mbE}{\mathbb{E}}

\newcommand{\mbI}{\mathbb{I}}

\newcommand{\mbO}{\mathbb{O}}

\newcommand{\mbS}{\mathbb{S}}

\DeclareMathOperator{\tr}{tr}
\DeclareMathOperator{\Tr}{Tr}
\DeclareMathOperator{\id}{id}

\newcommand{\ketbra}[2]{|{#1}\rangle\!\langle{#2}|}

\newcommand{\ba}{\begin{eqnarray}}
\newcommand{\ea}{\end{eqnarray}}
\newcommand{\bann}{\begin{eqnarray*}}
\newcommand{\eann}{\end{eqnarray*}}
\newcommand{\bal}{\begin{equation}\begin{aligned}}
\newcommand{\eal}{\end{aligned}\end{equation}}
\newcommand{\dm}[1]{\ketbra{#1}{#1}}

\newcolumntype{L}[1]{>{\raggedright}p{#1}}
\newcolumntype{C}[1]{>{\centering}p{#1}}
\newcolumntype{R}[1]{>{\raggedleft}p{#1}}
\newcolumntype{D}{>{\centering\arraybackslash}X}

\newcommand{\<}{\left\langle}
\renewcommand{\>}{\right\rangle}

\newcommand{\sbar}{\;\rule{0pt}{9.5pt}\right|\;}

\newcommand{\lset}{\left\{\left.}
\newcommand{\rset}{\right\}}

\DeclareMathOperator{\QEM}{QEM}
\DeclareMathOperator{\GHZ}{GHZ}

\begin{document}

\title{Universal Sampling Lower Bounds for Quantum Error Mitigation}

\author{Ryuji Takagi}
\email{ryuji.takagi@phys.c.u-tokyo.ac.jp}
\affiliation{Department of Basic Science, The University of Tokyo, Tokyo 153-8902, Japan}
\affiliation{Nanyang Quantum Hub, School of Physical and Mathematical Sciences, Nanyang Technological University, 637371, Singapore}

\author{Hiroyasu Tajima}
\email{hiroyasu.tajima@uec.ac.jp}
\affiliation{Department of Communication Engineering and Informatics, University of Electro-Communications, 1-5-1 Chofugaoka, Chofu, Tokyo, 182-8585, Japan}
\affiliation{JST, PRESTO, 4-1-8 Honcho, Kawaguchi, Saitama, 332-0012, Japan}

\author{Mile Gu}
\email{mgu@quantumcomplexity.org}
\affiliation{Nanyang Quantum Hub, School of Physical and Mathematical Sciences, Nanyang Technological University, 637371, Singapore}
\affiliation{Centre for Quantum Technologies, National University of Singapore, 3 Science Drive 2, 117543, Singapore}
\affiliation{MajuLab, CNRS-UNS-NUS-NTU International Joint Research Unit UMI 3654, Singapore}

\begin{abstract}
Numerous quantum error-mitigation protocols have been proposed, motivated by the critical need to suppress noise effects on intermediate-scale quantum devices. Yet, their general potential and limitations remain elusive. In particular, to understand the ultimate feasibility of quantum error mitigation, it is crucial to characterize the fundamental sampling cost --- how many times an arbitrary mitigation protocol must run a noisy quantum device. Here, we establish universal lower bounds on the sampling cost for quantum error mitigation to achieve the desired accuracy with high probability. Our bounds apply to general mitigation protocols, including the ones involving nonlinear postprocessing and those yet-to-be-discovered. The results imply that the sampling cost required for a wide class of protocols to mitigate errors must grow exponentially with the circuit depth for various noise models, revealing the fundamental obstacles in the scalability of useful noisy near-term quantum devices.
\end{abstract}


\maketitle


\textit{\textbf{Introduction}} ---
As recent technological developments have started to realize controllable small-scale quantum devices, a central problem in quantum information science has been to pin down what can and cannot be accomplished with noisy intermediate-scale quantum (NISQ) devices~\cite{Preskill2018quantum}.
One of the most relevant issues in understanding the ultimate capability of quantum hardware is to characterize how well noise effects could be circumvented. 
This is especially so for NISQ devices, as today's quantum devices generally cannot accommodate full quantum error correction that requires scalable quantum architecture.
As an alternative to quantum error correction, quantum error mitigation has recently attracted much attention as a potential tool to help NISQ devices realize useful applications~\cite{McArdle2020quantum,Endo2021hybrid}. 
It is thus of primary interest from practical and foundational viewpoints to understand the ultimate feasibility of quantum error mitigation.

Quantum error mitigation protocols generally involve running available noisy quantum devices many times. The collected data is then post-processed to infer classical information of interest.
While this avoids the engineering challenge in error correction, it comes at the price of \emph{sampling cost} --- computational overhead in having to sample a noisy device many times. 
This sampling cost represents the crucial quantity determining the feasibility of quantum error mitigation.
If the required sampling cost becomes too large, then such quantum error mitigation protocol becomes infeasible under a realistic time constraint. 
Various prominent quantum error mitigation methods face this problem, where sampling cost grows exponentially with circuit size~\cite{Yuan2016simulating,Temme2017error,Li2017efficient,Endo2018practical,Bravyi2021mitigating}. 
The crucial question then is whether there is hope to come up with a new error mitigation strategy that avoids this hurdle or if this is a universal feature shared by all quantum error mitigation protocols. 
To answer this question, we need a characterization of the sampling cost that is universally required for the general error-mitigation protocols, which has hitherto been unknown. 

Here, we provide a solution to this problem. We derive lower bounds for the number of samples fundamentally required for general quantum error mitigations to realize the target performance.
We then show that the required samples for a wide class of mitigation protocols to error-mitigate layered circuits under various noise models --- including the depolarizing and stochastic Pauli noise --- must grow exponentially with the circuit depth to achieve the target performance.  
This turns the conjecture that quantum error mitigation would generally suffer from the exponential sampling overhead into formal relations, extending the previous results on the exponential resource overhead required for noisy circuits without postprocessing~\cite{Unruh1995maintaining,Aharonov1996limitations,Palma1996quantum}.
We accomplish these by employing an information-theoretic approach, which establishes the novel connection between the state distinguishability and operationally motivated error-mitigation performance measures. 
Our results place the fundamental limitations imposed on the capability of general error-mitigation strategies that include existing protocols~\cite{Temme2017error,Li2017efficient,McClean2017hybrid,Bonet-Monroig2018low-cost,Endo2018practical,Mcclean2020decoding,Koczor2021exponential,Huggins2021virtual,Yoshioka2021generalized,Czarnik2021errormitigation,Strikis2021learning-based,Lowe2021unified,Czarnik2021qubit-efficient,O'Brien2021error,Huo2022dual-state,Bultrini2021unifying,Bravyi2021mitigating,Cai2021quantum,Wang2021mitigating,Cai2021multi-exponential,Mari2021extending,Mezher2022mitigating,Fontana2022spectral,Xiong2022quantum} and the ones yet to be discovered, being analogous to the performance converse bounds established in several other disciplines --- such as thermodynamics~\cite{Carnot1824reflections,Landauer1961irreversibility,Brandao2015second}, quantum communication~\cite{Bennett1999entanglement,Pirandola2017fundamental}, and quantum resource theories~\cite{Regula2021fundamental,Fang2022nogo} --- that contributed to characterizing the ultimate operational capability allowed in each physical setting.

Our work complements and extends several recent advancements in the field.
Ref.~\cite{Takagi2021fundamental} introduced a general framework of quantum error mitigation and established lower bounds for the maximum estimator spread, i.e., the range of the outcomes of the estimator, imposed on all error mitigation in the class, which provides a \emph{sufficient} number of samples to ensure the target accuracy. 
Those bounds were then employed to show that the maximum spread grows exponentially with the circuit depth to mitigate local depolarizing noise. 
Ref.~\cite{Wang2021can} showed a related result where for the class of error-mitigation strategies that only involve linear postprocessing, in which the target expectation value can be represented by a linear combination of the actually observed quantities, either the maximum estimator spread or the sample number needs to grow exponentially with the circuit depth to mitigate local depolarizing noise. 
The severe obstacle induced by noise in showing  a quantum advantage for variational quantum algorithms has also recently been studied~\cite{Wang2021noise-induced,Franca2021limitations,DePalma2023limitations}.
Our results lift the observations made in these works to rigorous bounds for the \emph{necessary} sampling cost required for general error mitigation, including the ones involving nonlinear postprocessing that constitute a large class of protocols~\cite{McClean2017hybrid,Bonet-Monroig2018low-cost,Endo2018practical,Mcclean2020decoding,Koczor2021exponential,Huggins2021virtual,Yoshioka2021generalized,Cai2021quantum,Cai2021multi-exponential,Mezher2022mitigating,Fontana2022spectral,Czarnik2021qubit-efficient,O'Brien2021error,Bultrini2021unifying,Nation2021scalable,Strikis2021learning-based,Xiong2022quantum}.


\textit{\textbf{Framework}} --- 
Suppose we wish to obtain the expectation value of an observable $A\in\mbO$ for an ideal state $\rho\in\mbS$ where $\mbO$ and $\mbS$ are some sets of observables and states. 
We assume that the ideal quantum state $\rho$ is produced by a unitary quantum circuit $\mU$ applied to the initial state $\rho_{\rm ini}\in\mbS_{\rm in}$ as $\rho=\mU(\rho_{\rm ini})$ where $\mbS_{\rm in}$ is the set of possible input states.
The noise in the circuit, however, prevents us from preparing the state $\rho$ exactly. 
We consider quantum error mitigation protocols that aim to estimate the true expectation value under the presence of noise in the following manner~\cite{Takagi2021fundamental} (see also Fig.~\ref{fig:framework}).

In the mitigation procedure, one can first modify the circuit, e.g., use a different choice of unitary gates with potential circuit simplification, apply nonadaptive operations (enabling, e.g., dynamical decoupling~\cite{Viola1998dynamical,Viola1999dynamical} and Pauli twirling~\cite{Li2017efficient}), and supply ancillary qubits --- the allowed modifications are determined by the capability of the available device.
Together with the noise present in the modified circuit, this turns the original unitary $\mU$ into some quantum channel $\mF$, which produces a \emph{distorted state} $\rho'$. 
The distorted state can be represented in terms of the ideal state $\rho$ by $\rho'=\mE(\rho)$ where we call $\mE\coloneqq \mF\circ\mU^\dagger$ an \emph{effective noise channel}.

The second step consists of collecting $N$ samples $\{\mE_n(\rho)\}_{n=1}^N$ of distorted states represented by a set of effective noise channels $\mbE\coloneqq\{\mE_n\}_{n=1}^N$ and
applying a \emph{trailing quantum process} $\mP_A$ over them.
The effective noise channels in $\mbE$ can be different from each other in general, as noisy hardware could have different noise profiles each time, or could purposely change the noise strength~\cite{Buscemi2013twopoint,Temme2017error}.
The trailing process $\mP_A$ then outputs an estimate represented by a random variable $\hat E_A(\rho)$ for the true expectation value $\Tr(A\rho)$.  
The main focus of our study is the \emph{sampling number} $N$, the total number $N$ of distorted states used in the error mitigation process.

We quantify the performance of an error-mitigation protocol by how well the protocol can estimate the expectation values for a given set $\mbO$ of observables and a set $\mbS$ of ideal states, which we call \emph{target observables} and \emph{target states} respectively.
We keep the choices of these sets general, and they can be flexibly chosen depending on one's interest.
For instance, if one is interested in error mitigation protocols designed to estimate the Pauli observables (e.g., virtual distillation~\cite{Koczor2021exponential,Huggins2021virtual,Czarnik2021qubit-efficient}), $\mbO$ can be chosen as the set of Pauli operators. 
As the trailing process includes a measurement depending on the observable, an error-mitigation strategy with target observables $\mbO$ is equipped with a family of trailing processes $\{\mP_A\}_{A\in\mbO}$.
Similarly, our results hold for an arbitrary choice of $\mbS$, where one can, for instance, choose this as the set of all quantum states, which better describes the protocols such as probabilistic error cancellation~\cite{Buscemi2013twopoint,Temme2017error,Takagi2021optimal,Regula2021operationalapplication,Jiang2021physical,Sun2021mitigating,Piveteau2022quasiprobability}, or as the set of states in a certain subspace, which captures the essence of subspace expansion~\cite{McClean2017hybrid,Mcclean2020decoding,Yoshioka2021generalized}.

This framework includes many error-mitigation protocols proposed so far~\cite{Temme2017error,Li2017efficient,McClean2017hybrid,Bonet-Monroig2018low-cost,Endo2018practical,Mcclean2020decoding,Koczor2021exponential,Huggins2021virtual,Yoshioka2021generalized,Czarnik2021errormitigation,Strikis2021learning-based,Lowe2021unified,Czarnik2021qubit-efficient,O'Brien2021error,Huo2022dual-state,Bultrini2021unifying,Bravyi2021mitigating,Cai2021quantum,Wang2021mitigating,Cai2021multi-exponential,Mari2021extending,Mezher2022mitigating,Fontana2022spectral,Xiong2022quantum}.
It is worth noting that our framework includes protocols that involve nonlinear postprocessing of the measurement outcomes.  
Error-mitigation protocols typically work by (1) making some set of (usually Pauli) measurements for observables $\{O_i\}_i$, (2) estimating their expectation values $\{\<O_i\>\}_i$ for distorted states, and (3) applying a classical postprocessing function $f$ over them.
The protocols with linear postprocessing functions, i.e., the ones with the form $f(\<O_i\>_i)=\sum_i c_i \<O_i\>$, are known to admit simpler analysis~\cite{Wang2021can,Takagi2021fundamental}, but numerous protocols --- including virtual distillation~\cite{Koczor2021exponential,Huggins2021virtual,Czarnik2021qubit-efficient}, symmetry verification~\cite{Bonet-Monroig2018low-cost}, and subspace expansion~\cite{McClean2017hybrid,Mcclean2020decoding,Yoshioka2021generalized} --- come with nonlinear postprocessing functions. 
In our framework, the sampling number $N$ is the \emph{total} number of samples used, where we consider the output represented by $\hat E_A(\rho)$ as our final guess and thus do not generally assume repeating some procedure many times and take a statistical average. This enables us to have any postprocessing absorbed in the trailing process $\mP_A$, making our results valid for the protocols with nonlinear postprocessing functions.

We also remark that our framework includes protocols with much more operational power than existing protocols, as we allow the trailing process to apply any coherent interaction over all distorted states.
Our results thus provide fundamental limits on the sampling overhead applicable to an arbitrary protocol in this extended class of error-mitigation protocols.

\begin{figure}
    \centering
    \includegraphics[width=\columnwidth]{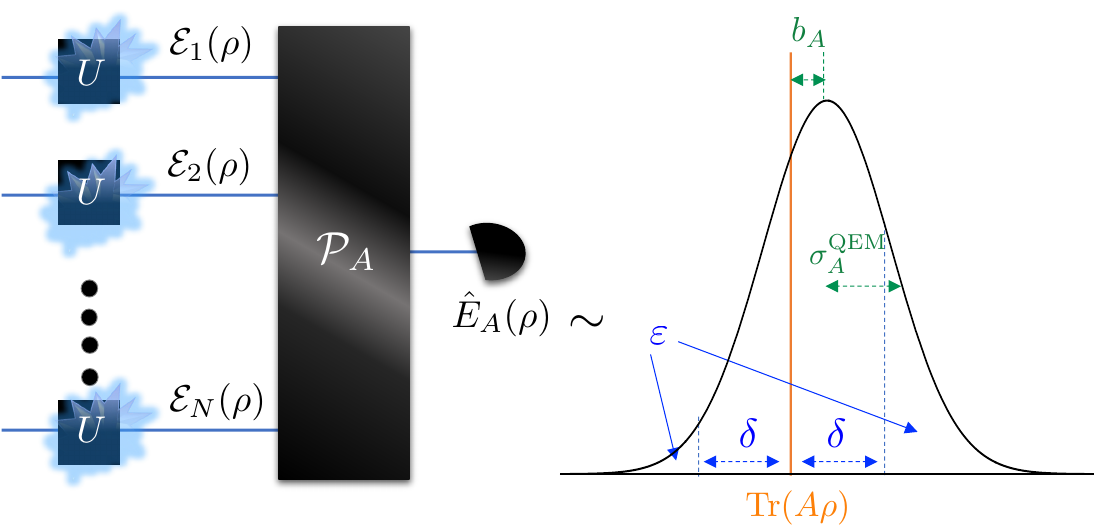}
    \caption{Framework of quantum error mitigation. For an ideal state $\rho\in\mbS$ and an observable $A\in\mbO$ of interest, we first prepare $N$ copies of distorted states $\{\mE_n(\rho)\}_{n=1}^N$, where $\mbE=\{\mE_n\}_{n=1}^N$ is  the set of effective noise channels. A trailing quantum process $\mP_A$ is then applied to $N$ distorted states, producing the final estimation of $\Tr(A\rho)$ represented by a random variable $\hat E_A(\rho)$. We quantify the error-mitigation performance in two ways by studying the property of the distribution of $\hat E_A(\rho)$; the first is the combination of the accuracy $\delta$ and the success probability $1-\varepsilon$, and the second is the combination of the bias $b_A(\rho)\coloneqq \<\hat E_A(\rho)\>-\Tr(A\rho)$ and the standard deviation $\sigma_A^{\QEM}$ of $\hat E_A(\rho)$.}
    \label{fig:framework}
\end{figure}


\textit{\textbf{Sampling lower bounds}} ---
We now consider the required samples to ensure the target performance. 
The performance of quantum error mitigation can be defined in multiple ways. 
Here, we consider two possible performance quantifiers that are operationally relevant. 

Our first performance measure is the combination of the accuracy of the estimate and the success probability.
This closely aligns with the operational motivation, where one would like an error mitigation strategy to be able to provide a good estimate for each observable in $\mbO$ and an ideal state in $\mbS$ at a high probability.  
This can be formalized as a condition 
\bal
 {\rm Prob}(|\Tr(A\rho)-\hat E_A(\rho)|\leq \delta)\geq 1-\varepsilon,\ \forall \rho\in\mbS,\,\forall A\in\mbO
 \label{eq:accuracy probability requirement}
\eal
where $\delta$ is the target accuracy and $1-\varepsilon$ is the success probability (see also Fig.~\ref{fig:framework}). 

The problem then is to identify lower bounds on the number $N$ of distorted states needed to achieve this condition as a function of $\delta$ and $\varepsilon$. 
We address this by observing that the trailing process of quantum error mitigation is represented as an application of a quantum channel and thus can never increase the state distinguishability. 
To formulate our result, let us define the observable-dependent distinguishability with respect to a set $\mbO$ of observables as 
\bal
D_\mbO(\rho,\sigma)\coloneqq \max_{A\in\mbO} |\Tr[A(\rho-\sigma)]|.
\label{eq:distinguishability observable dependent}
\eal
This quantity can be understood as the resolution in distinguishing two quantum states by using the measurements of the observables in $\mbO$.  
We note that when $\mbO=\{ A \,| 0\leq A \leq \mbI\}$ \footnote{Here, we consider the standard matrix inequality with respect to the positive semidefinite cone, i.e., $M_1\leq M_2$ for two matrices $M_1$ and $M_2$ means that the matrix $M_2-M_1$ is positive semidefinite.}, the quantity in \eqref{eq:distinguishability observable dependent} becomes the trace distance $D_{\tr}(\rho,\sigma)=\frac{1}{2}\|\rho-\sigma\|_1$~\cite{Nielsen2011quantum}. 

We then obtain the following sampling lower bounds applicable to an arbitrary given set $\mbE$ of effective noise channels. (Proof in Appendix~A~\footnote{See the Supplemental Material for detailed proofs and discussions of our main results, which includes Refs.~\cite{Fuchs1999cryptographic,Hiai1981sufficiency,Reeb2011Hibert,Kastoryano2013quantum,DePalma2021quantum,Gilchrist2005distance,Tajima2018uncertainty,Katsube2020fundamental,Hayashi2016quantum,Audenaert2005continuity,Carbone2015logarithmic,Kastoryano2016quantum,MullerHermes2016relative,MullerHermes2016entropy,Bardet2021onthemodified,Beigi2020quantum,Capel2020themodified,Hirche2020oncontraction,MullerLennert2013onquantum,Gorini1976completely,Temme2014hypercontractivity,Temme2015howfast,Capel2018quantum,Temme2010chi,qiskit_noise}}.)
\nocite{Fuchs1999cryptographic,Hiai1981sufficiency,Reeb2011Hibert,Kastoryano2013quantum,DePalma2021quantum,Gilchrist2005distance,Tajima2018uncertainty,Katsube2020fundamental,Hayashi2016quantum,Audenaert2005continuity,Carbone2015logarithmic,Kastoryano2016quantum,MullerHermes2016relative,MullerHermes2016entropy,Bardet2021onthemodified,Beigi2020quantum,Capel2020themodified,Hirche2020oncontraction,MullerLennert2013onquantum,Gorini1976completely,Temme2014hypercontractivity,Temme2015howfast,Capel2018quantum,Temme2010chi,qiskit_noise}

\begin{thm}\label{thm:sampling bound}
Suppose that an error-mitigation strategy achieves \eqref{eq:accuracy probability requirement} with some $\delta\geq 0$ and $0\leq \varepsilon\leq 1/2$ with $N$ distorted states characterized by the effective noise channels $\mbE=\{\mE_n\}_{n=1}^N$.
Then, the sample number $N$ is lower bounded as
\bal
 N &\geq \max_{\substack{\rho,\sigma\in\mbS\\D_\mbO(\rho,\sigma)\geq 2\delta}}\min_{\mE\in\mbE}\frac{\log\left[\frac{1}{4\varepsilon(1-\varepsilon)}\right]}{\log\left[ 1/F(\mE(\rho),\mE(\sigma))\right]},\\
N &\geq \max_{\substack{\rho,\sigma\in\mbS\\D_\mbO(\rho,\sigma)\geq 2\delta}}\min_{\mE\in\mbE}\frac{2(1-2\varepsilon)^2}{\ln 2 \cdot S(\mE(\rho)\|\mE(\sigma))},
 \label{eq:sampling lower bound general}
\eal
where $F(\rho,\sigma)\coloneqq \|\sqrt{\rho}\sqrt{\sigma}\|_1^2$ is the (square) fidelity and $S(\rho\|\sigma)\coloneqq \Tr(\rho\log \rho)-\Tr(\rho\log\sigma)$ is the relative entropy. 
\end{thm}

This result tells that if the noise effect brings states close to each other, it incurs an unavoidable sampling cost to error mitigation. 
The minimization over $\mbE$ chooses the effective noise channel that least reduces the infidelity and the relative entropy respectively.
On the other hand, the maximum over the ideal states represents the fact that to mitigate two states $\rho$ and $\sigma$ that are separated further than $2\delta$ in terms of observables in $\mbO$, the sample number $N$ that achieves the accuracy $\delta$ and the success probability $1-\varepsilon$ must satisfy the lower bounds with respect to $\rho$ and $\sigma$. 
The maximization over such $\rho$ and $\sigma$ then provides the tightest lower bound. 
This also reflects the observation that error mitigation accommodating a larger set $\mbO$ of target observables would require a larger number of samples.

We remark that although the set $\mbE$ --- which depends on how one modifies the noisy circuit --- ultimately depends on a specific error-mitigation strategy in mind, fixing $\mbE$ to a certain form already provides useful insights as we see later in the context of noisy layered circuits. We also stress that the above bounds hold for an arbitrary choice of $\mbE$, providing the general relation between the error mitigation performance and the information-theoretic quantity. 

The bounds in Theorem~\ref{thm:sampling bound} depend on the accuracy $\delta$ implicitly through the constraints on $\rho$ and $\sigma$ in the maximization.
For instance, if one sets $\delta=0$, one can find that both bounds diverge, as the choice of $\sigma=\rho$ would be allowed in the maximization.
In Appendix~B, we report an alternative bound that has an explicit dependence on the accuracy $\delta$.


Let us now consider our second performance measure based on the standard deviation and the bias of the estimate. 
Let $\sigma_A^{\QEM}(\rho)$ be the standard deviation of $\hat E_A(\rho)$ for an observable $A\in\mbO$, which represents the uncertainty of the final estimate of an error mitigation protocol. 
Since a good error mitigation protocol should come with a small fluctuation in its outcome, the standard deviation of the underlying distribution for the estimate can serve as a performance quantifier. 
However, the standard deviation itself is not sufficient to characterize the error mitigation performance, as one can easily come up with a useless strategy that always outputs a fixed outcome, which has zero standard deviation.  
This issue can be addressed by considering the deviation of the expected value of the estimate from the true expectation value called bias, defined as $b_A(\rho)\coloneqq \<\hat E_A(\rho)\>-\Tr(A\rho)$ for a state $\rho\in\mbS$ and an observable $A\in\mbO$ (see also Fig.~\ref{fig:framework}).

To assess the performance of error-mitigation protocols, we consider the worst-case error among possible ideal states and measurements. 
This motivates us to consider the maximum standard deviation $\sigma_{\max}^{\QEM}\coloneqq\max_{A\in\mbO}\max_{\rho\in\mbS}\sigma_A^{\QEM}(\rho)$ and the maximum bias $b_{\max}\coloneqq \max_{A\in\mbO}\max_{\rho\in\mbS} b_A(\rho)$.
Then, we obtain the following sampling lower bound in terms of these performance quantifiers. (Proof in Appendix~C.)

\begin{thm}\label{thm:sample bound standard deviation}
The sampling cost for an error-mitigation strategy with the maximum standard deviation $\sigma_{\max}^{\QEM}$ and the maximum bias $b_{\max}$ is lower bounded as 
\bal
 N\geq \max_{\substack{\rho,\sigma\in\mbS\\D_\mbO(\rho,\sigma)-2b_{\max}\geq 0}}\min_{\mE\in\mbE}\frac{\log\left[1-\frac{1}{\left(1+\frac{2\sigma_{\max}^{\QEM}}{D_\mbO(\rho,\sigma)-2b_{\max}}\right)^2}\right]^{-1}}{\log\left[ 1/F(\mE(\rho),\mE(\sigma))\right]}.
 \label{eq:bias standard deviation bound}
\eal
\end{thm}
This result represents the trade-off between the standard deviation, bias, and the required sampling cost. 
To realize the small standard deviation and bias, error mitigation needs to use many samples; in fact, the lower bound diverges at the limit of $\sigma_{\max}^{\QEM}\to 0$ whenever there exist states $\rho, \sigma\in\mbS$ such that $D_\mbO(\rho,\sigma)\geq 2b_{\max}$.  
On the other hand, a larger bias results in a smaller sampling lower bound, indicating a potential to reduce the sampling cost by giving up some bias.

The bounds in Theorems~\ref{thm:sampling bound}, \ref{thm:sample bound standard deviation} are universally applicable to arbitrary error mitigation protocols in our framework.
Therefore, our bounds are not expected to give good estimates for a given specific error-mitigation protocol in general, just as there is a huge gap between the Carnot efficiency and the efficiency of most of the practical heat engines.
Nevertheless, it is still insightful to investigate how our bounds are compared to existing mitigation protocols. 
In Appendix~D, we compare the bound in Theorem~\ref{thm:sampling bound} to the sampling cost for several error-mitigation methods, showing that our bound can provide nontrivial lower bounds with the gap being the factor of 3 to 6.
Although this does not guarantee that our bound behaves similarly for other scenarios in general, this ensures that there is a setting in which the bound in Theorem~\ref{thm:sampling bound} can provide a nearly tight estimate.  
We further show in Appendix~E that the scaling of the lower bound in Theorem~\ref{thm:sample bound standard deviation} with noise strength can be achieved by the probabilistic error cancellation method in a certain scenario.
This shows that probabilistic error cancellation serves as an optimal protocol in this specific sense, complementing the recent observation on the optimality of probabilistic error cancellation established for the maximum estimator spread measure~\cite{Takagi2021fundamental}.


\textit{\textbf{Noisy layered circuits}} ---
The above results clarify the close relation between the sampling cost and state distinguishability. 
As an application of our general bounds, we study the inevitable sample overhead to mitigate noise in the circuits consisting of multiple layers of unitaries. 
Although we here focus on the local depolarizing noise, our results can be extended to a number of other noise models as we discuss later.

Suppose that an $M$-qubit quantum circuit consists of layers of unitaries, each of which is followed by a local depolarizing noise, i.e., a depolarizing noise of the form $\mD_p=(1-p)\id + p\mbI/2$ where $p$ is a noise strength, applies to each qubit.
We aim to estimate ideal expectation values for the target states $\mbS$ and observables $\mbO$ by using $N$ such noisy layered circuits.
Although the noise strength can vary for different locations, we suppose that $L$ layers are followed by the local depolarizing noise with noise strength of at least $\gamma$.
We call these layers $U_1,U_2,\dots, U_L$ and let $\gamma_{n,l,m}$ denote the noise strength of the local depolarizing noise on the $m^{\rm th}$ qubit after the $l^{\rm th}$ unitary layer $U_l$ in the $n^{\rm th}$ noisy circuit, where $m\leq M,\,l\leq L,\, n\leq N$. 
This gives the expression of the local depolarizing noise after $l^{\rm th}$ layer in the $n^{\rm th}$ noisy circuit as $\otimes_{m=1}^M \mD_{\gamma_{n,l,m}}$, where $\gamma_{n,l,m}\geq \gamma\,\forall n,l,m$.

Here, we focus on the error-mitigation protocols that apply an arbitrary trailing process over $N$ distorted states and any unital operations (i.e., operations that preserve the maximally mixed state) before and after $U_l$ (Fig.~\ref{fig:noisy_layered}). 
This structure ensures that error correction does not come into play here, as the size of input and output spaces of the intermediate unital channels is restricted to $M$ qubits, as well as that unital channels do not serve as good decoders for error correction.

We show that the necessary number of samples required to achieve the target performance grows exponentially with the number of layers in both performance quantifiers introduced above.

\begin{thm}\label{thm:layer}
Suppose that an error-mitigation strategy described above is applied to an $M$-qubit circuit to mitigate local depolarizing channels with strength at least $\gamma$ that follow $L$ layers of unitaries, and achieves \eqref{eq:accuracy probability requirement} with some $\delta\geq 0$ and $0\leq \varepsilon\leq 1/2$.
Then, if there exist at least two states $\rho,\sigma\in\mbS$ such that $D_\mbO(\rho,\sigma)\geq 2\delta$, the required sample number $N$ is lower bounded as
\bal
 N\geq \frac{(1-2\varepsilon)^2}{2\ln (2)\,M(1-\gamma)^{2L}}.
\eal
\end{thm}

The proof can be found in Appendix~F.
This result particularly shows that the required number of samples must grow exponentially with the circuit depth $L$.
We remark that the bound always holds under the mild condition, i.e., $D_\mbO(\rho,\sigma)\geq 2\delta$ for some $\rho,\sigma\in\mbS$. 
This reflects that, to achieve the desired accuracy $\delta$ satisfying this condition, error mitigation really needs to extract the expectation values about the observables in $\mbO$ and the states in $\mbS$, prohibiting it from merely making a random guess.

In Appendix~G, we obtain a similar exponential growth of the required sample overhead for a fixed target bias and standard deviation. 
We also obtain in Appendix~H alternative bounds that are tighter in the range of small $\varepsilon$.

With a suitable modification of allowed unitaries and intermediate operations, we extend these results to a wide class of noise models, including stochastic Pauli, global depolarizing, and thermal noise.
The case of thermal noise particularly provides an intriguing physical interpretation: the sampling cost $N$ required to mitigate thermal noise after time $t$ is characterized by the loss of free energy $N=\Omega(1/[F(\rho_t)-F_{\rm eq}])$ where $\rho_t$ is the state at time $t$ and $F_{\rm eq}$ is the equilibrium free energy. 
This in turn shows that the necessary sampling cost grows as $N=\Omega(e^{\alpha_{\rm ent} t})$ where $\alpha_{\rm ent}$ is a constant characterized by the minimum entropy production rate. 
We provide details on these extensions in Appendix~I.

We remark that Theorem~\ref{thm:layer} (and related results discussed in the Appendices) extends the previous results showing the exponential resource overhead required for noisy circuits without postprocessing~\cite{Unruh1995maintaining,Aharonov1996limitations,Palma1996quantum}. In Appendix~J, we provide further clarifications about the differences between the settings considered in the previous works and ours.

\begin{figure}
    \centering
    \includegraphics[width=\columnwidth]{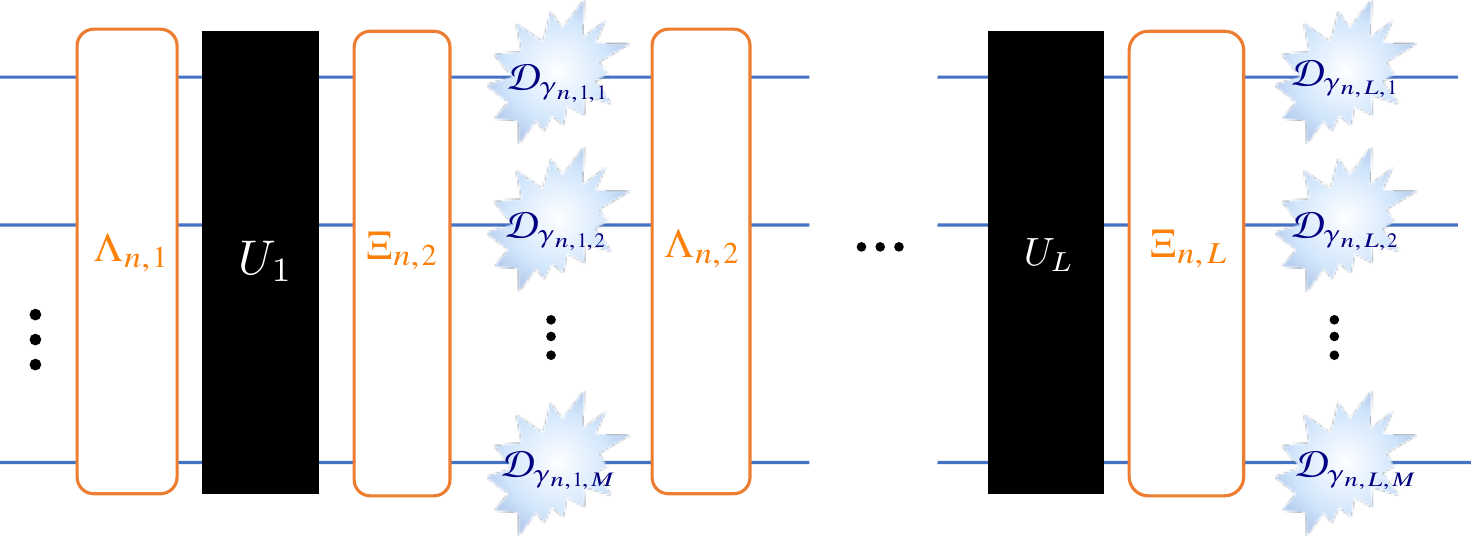}
    \caption{Each distorted state ($n^{\rm th}$ copy depicted in the figure) is produced by a circuit with $L$ layers $U_1,\dots,U_L$ followed by a local depolarizing noise with noise strength at least $\gamma$, i.e., $\gamma_{n,l,m}\geq \gamma, \forall n,l,m$. Each layer $U_l$ can be sandwiched by additional unital operations $\Lambda_{n,l}$ and $\Xi_{n,l}$. Other layers and depolarizing channels with noise strength smaller than $\gamma$ are absorbed in these operations as they are also unital.}
    \label{fig:noisy_layered}
\end{figure}


\textit{\textbf{Conclusions}} --- 
We established sampling lower bounds imposed on the general quantum error-mitigation protocols. Our results formalize the idea that the reduction in the state distinguishability caused by noise and error-mitigation processes leads to the unavoidable computational overhead in quantum error mitigation.
We then showed that error-mitigation protocols with certain intermediate operations and an arbitrary trailing process require the number of samples that grows exponentially with the circuit depth to mitigate various types of noise. 
We presented these bounds with respect to multiple performance quantifiers --- accuracy and success probability, as well as the standard deviation and bias --- each of which has its own operational relevance.  

Our bounds provide fundamental limitations that universally apply to general mitigation protocols, clarifying the underlying principle that regulates error-mitigation performance.
As a trade-off, they may not give tight estimates for a given specific error-mitigation strategy, analogously to many other converse bounds established in other fields that typically give loose bounds for most specific protocols.  
A thorough study to identify in what setting our bounds can give good estimates will make an interesting future research direction.


\begin{acknowledgments}

\textit{Acknowledgments} --- We thank Suguru Endo, Shintaro Minagawa, Masato Koashi, and Daniel Stilck Fran\c{c}a for helpful discussions, and Kento Tsubouchi, Takahiro Sagawa, and Nobuyuki Yoshioka for sharing a preliminary version of their manuscript. We acknowledge the support of
the Singapore Ministry of Education Tier 1 Grants
RG146/20 and RG77/22 (S), the NRF2021-QEP2-
02-P06 from the Singapore Research Foundation and
the Singapore Ministry of Education Tier 2 Grant
T2EP50221-0014 and the FQXi R-710-000-146-720
Grant “Are quantum agents more energetically efficient
at making predictions?” from the Foundational Questions Institute and Fetzer Franklin Fund (a donor-advised
fund of Silicon Valley Community Foundation). R.T. was also supported by the Lee Kuan Yew Postdoctoral Fellowship at Nanyang Technological University Singapore. H.T. is supported by JSPS Grants-in-Aid for Scientific Research No. JP19K14610 and No. JP22H05250, JST PRESTO No. JPMJPR2014, and JST MOONSHOT No. JPMJMS2061.
\end{acknowledgments}


\textit{Note added.} --- During the completion of our manuscript, we became aware of an independent work by Tsubouchi \emph{et al.}~\cite{Tsubouchi2022} that obtained a result related to our Theorem~S.2 in Appendix~G, in which they showed an alternative exponential sample lower bound applicable to error-mitigation protocols that achieve zero bias using quantum estimation theory.


\bibliographystyle{apsrmp4-2}
\bibliography{myref}


\clearpage
\newgeometry{hmargin=1.2in,vmargin=0.8in}

\widetext

\begin{center}
{\large \bf  Universal Sampling Lower Bounds for Quantum Error Mitigation\\ \vspace{0.2cm} Supplemental Material}
\end{center}

\appendix

\setcounter{thm}{0}
\renewcommand{\thethm}{S.\arabic{thm}}
\setcounter{figure}{0}
\renewcommand{\thefigure}{S.\arabic{figure}}

\section{Proof of Theorem~\ref{thm:sampling bound}}\label{app:sampling bound proof}

\begin{proof}
 For an ideal state $\rho\in\mbS$ and an observable $A\in\mbO$, let $p_{A,\rho}$ be the probability distribution for $\hat E_A(\rho)$.
 This means that we have 
 \bal
  \mP_A\left(\otimes_{n=1}^N \mE_n(\rho)\right)=\int da p_{A,\rho}(a)\dm{a}
  \label{eq:mitigation channel}
 \eal  
 where $\{\ket{a}\}$ is the classical states representing possible estimates represented by the random variable $\hat E_A(\rho)$. 
 
 Take arbitrary two states $\rho, \sigma\in\mbS$ satisfying $|\Tr(A\rho)-\Tr(A\sigma)|\geq 2\delta$.
 We assume $\Tr(A\rho)\geq \Tr(A\sigma)$ without loss of generality. 
 Let us divide the regions of estimates into three sections as (see Fig.~\ref{fig:overlap})
 \bal
  L\coloneqq \lset a \sbar a\leq \Tr(A\sigma)+\delta\rset,\ M\coloneqq \lset a \sbar \Tr(A\sigma)+\delta < a < \Tr(A\rho)-\delta\rset,\ R\coloneqq \lset a \sbar a\geq \Tr(A\rho)-\delta\rset.
 \eal

\begin{figure}[b]
    \centering
    \includegraphics[width=.5\columnwidth]{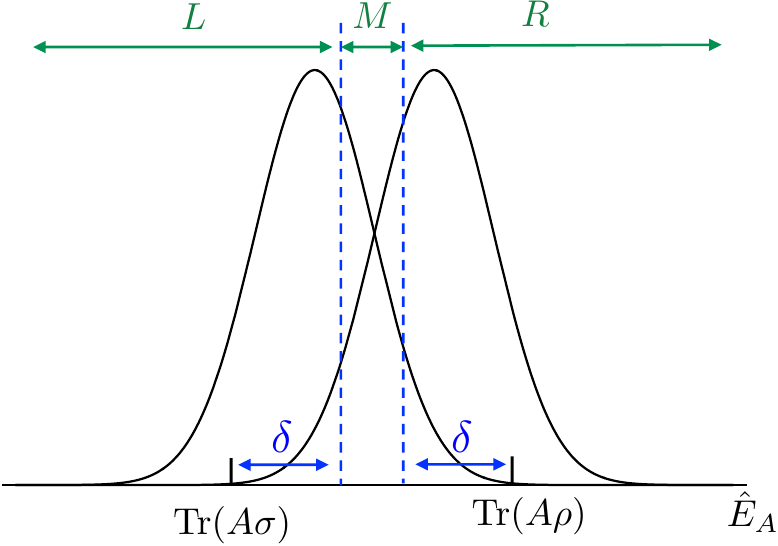}
    \caption{Probability distributions of $\hat E_A(\rho)$ and $\hat E_A(\sigma)$ for some states $\rho,\sigma\in\mbS$. The condition \eqref{eq:accuracy probability requirement} requires that $1-\varepsilon$ of the probability must be present around the true expectation values with $\delta$ deviation for both $\rho$ and $\sigma$.}
    \label{fig:overlap}
\end{figure}

The condition \eqref{eq:accuracy probability requirement} ensures that 
\bal
 \int_L da\,p_{A,\rho}(a)\leq \varepsilon,\quad\int_R da\, p_{A,\sigma}(a)\leq \varepsilon,\\
 \int_R da\,p_{A,\rho}(a)\geq 1-\varepsilon,\quad\int_L da\, p_{A,\sigma}(a)\geq 1-\varepsilon.
\label{eq:prob bound}
\eal
This allows us to bound the trace distance between two classical probability distributions $p_{A,\rho}$ and $p_{A,\sigma}$ as 
\bal
 D_{\rm tr}(p_{A,\rho},p_{A,\sigma})&=\frac{1}{2}\int_{L+M+R} da\,  |p_{A,\rho}(a)-p_{A,\sigma}(a)|\\
 &\geq \frac{1}{2}\int_L da\,  |p_{A,\rho}(a)-p_{A,\sigma}(a)|+\frac{1}{2}\int_R da\,  |p_{A,\rho}(a)-p_{A,\sigma}(a)|\\
 &\geq \frac{1}{2}\int_L da\,\left[  p_{A,\sigma}(a)-p_{A,\rho}(a)\right]+\frac{1}{2}\int_R da\,\left[  p_{A,\rho}(a)-p_{A,\sigma}(a)\right]\\
 &\geq 1-2\varepsilon
\eal
where we used \eqref{eq:prob bound} in the last inequality.
Noting that $\mP_A$ in \eqref{eq:mitigation channel} is a quantum channel, we apply the data-processing inequality to the trace distance to get 
\bal
 D_{\rm tr}(\otimes_{n=1}^N\mE_n(\rho),\otimes_{n=1}^N\mE_n(\sigma))\geq 1-2\varepsilon.
\label{eq:trace distance bound}
\eal
To extract $N$ out, recall that $D_{\rm tr}(\rho,\sigma)\leq \sqrt{1-F(\rho,\sigma)}$~\cite{Fuchs1999cryptographic} and that $F(\otimes_i \rho_i,\otimes_i\sigma_i)=\prod_i F(\rho_i,\sigma_i)$. 
Under the condition $\varepsilon\leq 1/2$, we can further bound \eqref{eq:trace distance bound} as 
\bal
 1-\prod_{n=1}^N F(\mE_n(\rho),\mE_n(\sigma))\geq (1-2\varepsilon)^2,
\eal
leading to 
\bal
 \sum_{n=1}^N\log F(\mE_n(\rho),\mE_n(\sigma))^{-1}\geq \log\left[\frac{1}{4\varepsilon(1-\varepsilon)}\right].
\eal
Taking $n^\star\coloneqq {\rm argmin}_n F(\mE_n(\rho),\mE_n(\sigma))$, i.e., $F(\mE_{n^\star}(\rho),\mE_{n^\star}(\sigma))^{-1}\geq F(\mE_{i}(\rho),\mE_{i}(\sigma))^{-1}, \forall i$, we get
\bal
 N\geq \frac{\log\left[\frac{1}{4\varepsilon(1-\varepsilon)}\right]}{\log F(\mE_{n^\star}(\rho),\mE_{n^\star}(\sigma))^{-1}}.
 \label{eq:fidelity bound}
\eal
Since this holds for every state $\rho$ and $\sigma$ such that $|\Tr(A\rho)-\Tr(A\sigma)|\geq 2\delta$, we maximize the right-hand side over such states to reach the advertised statement.

The relation with respect to the relative entropy can be obtained by upper bounding \eqref{eq:trace distance bound} with the relative entropy. 
Using the quantum Pinsker's inequality~\cite{Hiai1981sufficiency}
\bal
 D_{\tr}(\rho,\sigma)\leq \sqrt{\frac{\ln 2}{2}}\sqrt{S(\rho\|\sigma)}
\eal
and the additivity $S(\otimes_i\rho_i\|\otimes_i\sigma_i)=\sum_i S(\rho_i\|\sigma_i)$, we can bound \eqref{eq:trace distance bound} as 
\bal
 \frac{\ln 2}{2}\sum_{n=1}^N S\left(\mE_n(\rho)\|\mE_n(\sigma)\right)\geq (1-2\varepsilon)^2
\eal
when $\varepsilon\leq 1/2$.
Taking $n^\star\coloneqq {\rm argmax}_n S\left(\mE_n(\rho)\|\mE_n(\sigma)\right)$, i.e., $S(\mE_{n^\star}(\rho)\|\mE_{n^\star}(\sigma))\geq S(\mE_{i}(\rho)\|\mE_{i}(\sigma)), \forall i$, we get
\bal
 N\geq \frac{2(1-2\varepsilon)^2}{\ln 2 \cdot S(\mE_{n^\star}(\rho)\|\mE_{n^\star}(\sigma))}.
 \label{eq:relative entropy bound}
\eal
Since this holds for every state $\rho$ and $\sigma$ such that $|\Tr(A\rho)-\Tr(A\sigma)|\geq 2\delta$, we maximize the right-hand side over such states to reach the advertised statement.

\end{proof}

\section{Alternative bound with an explicit accuracy dependence}\label{app:bound with explicit accuracy dependence}

Here, we obtain an alternative bound to the one in Theorem~\ref{thm:sampling bound} that has an explicit dependence on the accuracy $\delta$ when $\mbS$ contains all pure states.
To this end, let us define a generalized contraction coefficient with respect to $\mbO$ by
\bal
 \eta^{\mbE,\mbO} \coloneqq \max_{\mE\in\mbE}\max_{\rho,\sigma\in\mbD} \frac{D_{\tr}(\mE(\rho),\mE(\sigma))}{D_\mbO(\rho,\sigma)},
\eal
where $\mbD$ is an arbitrary set that contains all pure quantum states. 
Note that $\eta^{\mbE,\mbO}$ can take any nonnegative value in general.
However, this particularly becomes the measure of contraction of trace distance when $\mbO=\{ A \,|\,0\leq A \leq \mbI \}$, in which case $0\leq\eta^{\mbE,\mbO}\leq 1$ always holds~\cite{Reeb2011Hibert,Kastoryano2013quantum,DePalma2021quantum}. 
The following result represents the lower bound in terms of accuracy, success probability, and the generalized contraction coefficient. 

\begin{pro}\label{pro:sample lower bound contraction}
Suppose $\mbS$ contains the set of all pure states. Then, the sampling cost $N$ that satisfies \eqref{eq:accuracy probability requirement} is lower bounded as 
\bal
 N\geq \frac{\log\left[\frac{1}{4\varepsilon(1-\varepsilon)}\right]}{\log\frac{1}{\left(1-2\eta^{\mbE,\mbO}\,\delta\right)^2}}.
\eal
\end{pro}

In the case of $\varepsilon\ll 1$ and $\delta\ll 1$, the lower bound approximately becomes $\log\left(\frac{1}{4\varepsilon}\right)/(4\eta^{\mbE,\mbO}\,\delta)$.
This bound has  the advantage of having explicit dependence on the accuracy and separating the contraction coefficient of the effective noise channel, making explicit the role of the reduction in the state distinguishability.

\begin{proof}
Let $\eta^{\mbE,\mbO,\mbS}$ be the generalized contraction coefficient (c.f., Refs.~\cite{Reeb2011Hibert,Kastoryano2013quantum,DePalma2021quantum}) for $\mE$ defined by 
\bal
 \eta^{\mbE,\mbO,\mbS} \coloneqq \max_{\mE\in\mbE}\max_{\rho,\sigma\in\mbS} \frac{D_{\tr}(\mE(\rho),\mE(\sigma))}{D_\mbO(\rho,\sigma)}.
\eal
Noting that $F(\rho,\sigma)\geq \left[1-D_{\tr}(\rho,\sigma)\right]^2$, we have for arbitrary $\rho,\sigma\in\mbS$ that
\bal
 \max_{\mE\in\mbE} \log F(\mE(\rho),\mE(\sigma))^{-1}&\leq \log\frac{1}{\left[1-\max_{\mE\in\mbE}D_{\tr}(\mE(\rho),\mE(\sigma))\right]^2}\\
 &\leq \log\frac{1}{\left[1-\eta^{\mbE,\mbO,\mbS}\max_{A\in\mbO}|\Tr[A(\rho-\sigma)]|\right]^2} 
\eal

Let $\tilde\delta\coloneqq \min_{\rho,\sigma\in\mbS}\lset \delta' \sbar \delta'=\max_{A\in\mbO}|\Tr[A(\rho-\sigma)]|\geq 2\delta \rset$. 
Maximizing the lower bound in \eqref{eq:sampling lower bound general} over $\rho,\sigma\in\mbS$ such that $\max_{A\in\mbO}|\Tr(A\rho)-\Tr(A\sigma)|\geq 2\delta$, we get 
\bal
 N\geq \frac{\log\left[\frac{1}{4\varepsilon(1-\varepsilon)}\right]}{\log\frac{1}{\left(1-\eta^{\mbE,\mbO,\mbS}\,\tilde\delta\right)^2}}.
\eal

Now, let $\mbD$ be a set that contains all pure states and consider arbitrary error mitigation with $\mbS=\mbD$. 
In such a scenario, we have $\tilde \delta = 2\delta$ and $\eta^{\mbE,\mbO,\mbS}=\eta^{\mbE,\mbO}$ where 
\bal
 \eta^{\mbE,\mbO} \coloneqq \max_{\mE\in\mbE}\max_{\rho,\sigma\in\mbD} \frac{D_{\tr}(\mE(\rho),\mE(\sigma))}{D_{\mbO}(\rho,\sigma)}.
\eal
This gives
\bal
 N\geq \frac{\log\left[\frac{1}{4\varepsilon(1-\varepsilon)}\right]}{\log\frac{1}{\left(1-2\eta^{\mbE}\,\delta\right)^2}},
\eal
concluding the proof.

\end{proof}

\section{Proof of Theorem~\ref{thm:sample bound standard deviation}}\label{app:sample bound standard deviation proof}

\begin{proof}

Let $D_F(\rho,\sigma)\coloneqq \sqrt{1-F(\rho,\sigma)}$ be the purified distance~\cite{Gilchrist2005distance}. 
For an observable $A$, let $\sigma_A(\rho)\coloneqq \sqrt{\Tr[(A-\Tr[A\rho])^2\rho]}$ be the standard deviation of the probability distribution for measuring the observable $A$ for state $\rho$.
Then, as an improvement of a relation reported in Ref.~\cite{Tajima2018uncertainty}, it was shown in Ref.~\cite{Katsube2020fundamental} that arbitrary states $\eta,\tau$ and an observable $O$ satisfy
\bal
 |\Tr[O(\eta-\tau)]|\leq D_F(\eta,\tau)\left(\sigma_O(\eta)+\sigma_O(\tau)+|\Tr[O(\eta-\tau)]|\right).
\label{eq:distance fluctuation}
\eal

Using the data-processing inequality and applying \eqref{eq:distance fluctuation} to the error-mitigated classical states, we get for arbitrary $\rho,\sigma\in\mbS$ that 
\bal
 D_F(\otimes_{n=1}^N\mE_n(\rho),\otimes_{n=1}^N\mE_n(\sigma))&\geq D_F(\mP_A\circ\otimes_{n=1}^N\mE_n(\rho),\mP_A\circ\otimes_{n=1}^N\mE_n(\sigma))\\
 &\geq \frac{\Delta}{\sigma_A^{\QEM}(\rho)+\sigma_A^{\QEM}(\sigma)+\Delta}.
\eal
Here, $\Delta\coloneqq |\Tr[A'\,\mP_A\circ\otimes_{n=1}^N\mE_n(\rho)]-\Tr[A'\,\mP_A\circ\otimes_{n=1}^N\mE_n(\sigma)]| = |\<\hat E_A(\rho)\>-\<\hat E_A(\sigma)\>|$ where $A'=\sum_a \dm{a}$ is the observable corresponding to the estimate (c.f., \eqref{eq:mitigation channel}), and
$\sigma_A^{\QEM}(\rho)\coloneqq\sigma_{A'}(\mP_A\circ\otimes_{n=1}^N\mE_n(\rho))$ is the standard deviation of $\hat E_A(\rho)$.
Reorganizing the terms gives 
\bal
 \sigma_A^{\QEM}(\rho)+\sigma_A^{\QEM}(\sigma)\geq \left(\frac{1}{D_F(\otimes_{n=1}^N\mE_n(\rho),\otimes_{n=1}^N\mE_n(\sigma))}-1\right)\Delta
\eal
Without loss of generality, let us take $\<\hat E_A(\rho)\>-\<\hat E_A(\sigma)\>\geq 0$. Then, recalling the definition of bias $b_A(\rho)\coloneqq \<\hat E_A(\rho)\>-\Tr[A\rho]$, we have 
\bal
 \Delta = b_A(\rho) - b_A(\sigma) + \Tr[A(\rho-\sigma)].
\eal
Combining these, we get
\bal
 \sigma_A^{\QEM}(\rho)+\sigma_A^{\QEM}(\sigma)\geq \left(\frac{1}{D_F(\otimes_{n=1}^N\mE_n(\rho),\otimes_{n=1}^N\mE_n(\sigma))}-1\right)(\Tr[A(\rho-\sigma)]+b_A(\rho) - b_A(\sigma)).
\eal

Defining he maximum standard deviation as $\sigma_{\max}^{\QEM}\coloneqq \max_{A\in\mbO}\max_{\rho\in\mbS}\sigma_{A,\max}^{\QEM}$ and the maximum bias $b_{\max}\coloneqq \max_{A\in\mbO}\max_{\rho\in\mbS}|b_A(\rho)|$, we get 
\bal
\sigma_{\max}^{\QEM}\geq \max_{\rho,\sigma\in\mbS}\frac{1}{2}\left(\frac{1}{D_F(\otimes_{n=1}^N\mE_n(\rho),\otimes_{n=1}^N\mE_n(\sigma))}-1\right)(D_\mbO(\rho,\sigma)-2b_{\max}).
\label{eq:standard deviation lower bound observable independent}
\eal
where $D_{\mbO}$ is the observable-dependent distinguishability defined in \eqref{eq:distinguishability observable dependent}.

The bound \eqref{eq:standard deviation lower bound observable independent} can be turned into a lower bound on $N$.
To see this, note that \eqref{eq:standard deviation lower bound observable independent} ensures that for arbitrary $\rho,\sigma\in\mbS$, we have 
\bal
\sigma_{\max}^{\QEM}\geq \frac{1-D_F(\otimes_{n=1}^N\mE_n(\rho),\otimes_{n=1}^N\mE_n(\sigma))}{2D_F(\otimes_{n=1}^N\mE_n(\rho),\otimes_{n=1}^N\mE_n(\sigma))}(D_\mbO(\rho,\sigma)-2b_{\max}),
\eal
which leads to 

\bal
D_F(\otimes_{n=1}^N\mE_n(\rho),\otimes_{n=1}^N\mE_n(\sigma))&\geq \frac{D_\mbO(\rho,\sigma)-2b_{\max}}{2\sigma_{\max}^{\QEM}+D_\mbO(\rho,\sigma)-2b_{\max}
}\\
&=\frac{1}{1+\frac{2\sigma_{\max}^{\QEM}}{D_\mbO(\rho,\sigma)-2b_{\max}}
}.
\eal

Noting that 
\bal
D_F(\otimes_{n=1}^N\mE_n(\rho),\otimes_{n=1}^N\mE_n(\sigma))&=\sqrt{1-F(\otimes_{n=1}^N\mE_n(\rho),\otimes_{n=1}^N\mE_n(\sigma))}\\
&=\sqrt{1-\prod_{n=1}^N F(\mE_n(\rho),\mE_n(\sigma))}
\eal
where we used the multiplicativity of the fidelity for tensor-product states. 
This gives 
\bal
\prod_{n=1}^N F(\mE_n(\rho),\mE_n(\sigma))\leq 1-\frac{1}{\left(1+\frac{2\sigma_{\max}^{\QEM}}{D_\mbO(\rho,\sigma)-2b_{\max}}\right)^2
}
\eal
for $\rho,\sigma\in\mbS$ such that $D_\mbO(\rho,\sigma)-2b_{\max}\geq 0$.
Taking the inverse and logarithm on both sides, we get
\bal
\sum_{n=1}^N \log F^{-1}(\mE_n(\rho),\mE_n(\sigma))\geq \log\left[1-\frac{1}{\left(1+\frac{2\sigma_{\max}^{\QEM}}{D_\mbO(\rho,\sigma)-2b_{\max}}\right)^2
}\right]^{-1}.
\eal
Noting 
\bal
N\max_{\mE\in\mbE}\log F^{-1}(\mE(\rho),\mE(\sigma))\geq \sum_{n=1}^N \log F^{-1}(\mE_n(\rho),\mE_n(\sigma)),
\eal
we reach 
\bal
 N \geq \min_{\mE\in\mbE}\frac{\log\left[1-\frac{1}{\left(1+\frac{2\sigma_{\max}^{\QEM}}{D_\mbO(\rho,\sigma)-2b_{\max}}\right)^2
}\right]^{-1}}{\log F^{-1}(\mE(\rho),\mE(\sigma))}
\label{eq:sample lower bound purified app}
\eal
for arbitrary states $\rho,\sigma\in\mbS$ such that $D_\mbO(\rho,\sigma)-2b_{\max}\geq 0$. 

\end{proof}


\section{Comparison of the bound in Theorem~\ref{thm:sampling bound} with existing protocols} \label{app:numerics accuracy probability bound}

To make an explicit comparison between the bound in Theorem~\ref{thm:sampling bound} and the sample cost for specific protocols, we consider mitigating the local depolarizing noise $\mD_p^{\otimes M}$, where $\mD_p(\rho)\coloneqq (1-p)\rho + p\mbI/2$, with accuracy $\delta$ and probability $1-\varepsilon$.
For the sake of the analysis, we consider conservative strategies with the target observable $\mbO=\{\prod_{i=1}^M X_i\}$ and target ideal states $\mbS=\{\GHZ(\pm\theta_1),\GHZ(\pm\theta_2)\}$ where
\bal
 \ket{\GHZ(\theta)}\coloneqq \frac{1}{\sqrt{2}}\left(\ket{0}^{\otimes M} + e^{i(\theta+\pi/2)}\ket{1}^{\otimes M}\right)
\eal
and we set $\theta_1 \coloneqq \arcsin(\delta)$, $\theta_2 \coloneqq \arcsin(2\delta)$.
This refers to the setting in which one would like to obtain good estimates at least for these states and observable. 
We also set the effective noise channels being the local depolarizing noise, i.e., $\mE_i=\mD_p^{\otimes M},\ \forall i$, which is satisfied for all the specific protocols we study below.

Since $\Tr[\prod_{i=1}^M X_i \GHZ(\theta)]= -\sin\theta$, we particularly have $D_\mbO(\GHZ(\theta_1),\GHZ(-\theta_1))=2\delta$. This allows us to bound the lower bound in Theorem~\ref{thm:sampling bound} to get 
\bal
 N\geq \frac{\log\frac{1}{4\varepsilon(1-\varepsilon)}}{\log[1/F(\mD_p^{\otimes M}(\GHZ(\theta_1)),D_p^{\otimes M}(\GHZ(-\theta_1)))]}.
 \label{eq:lower bound local depolarizing}
\eal

To study the sample numbers for specific error-mitigation protocols, we simulate each mitigation strategy and evaluate the actual sample numbers used to achieve the desired performance. 
Namely, we evaluate the probability of getting estimates within the accuracy $\delta$ for all ideal states in $\mbS$ using a certain sample budget. 
If this success probability is above the target probability $1-\varepsilon$, we need less number of samples. On the other hand, if the success probability is below $1-\varepsilon$, we need more samples.
By binary search, we evaluate the sample number that realizes the success probability $1-\varepsilon$ with a small error tolerance, which we set 0.01.
We also note that since the shift of the expectation value due to depolarizing noise is larger for $\GHZ(\pm\theta_2)$ than $\GHZ(\pm\theta_1)$ and the noise and error mitigation apply to $\GHZ(\theta_2)$ and $\GHZ(-\theta_2)$ in a symmetric manner, it suffices to ensure that the expectation value for $\GHZ(\theta_2)$ can be estimated with the accuracy $\delta$ with probability $1-\varepsilon$.  

The error-mitigation protocols we study are virtual distillation, symmetry verification, and probabilistic error cancellation.
For virtual distillation, we consider 2-copy virtual distillation and estimate $\Tr(\prod_{i=1}^M X_i[\mD_p^{\otimes M}(\GHZ(\theta_2))]^2)/\Tr([\mD_p^{\otimes M}(\GHZ(\theta_2))]^2)$ as in Refs.~\cite{Koczor2021exponential,Huggins2021virtual}.
For symmetry verification~\cite{Bonet-Monroig2018low-cost}, we observe that all the ideal states in $\mbS$ are +1 eigenstates of the stabilizers generated by $\{Z_iZ_{i+1}\}_{i=1}^{M-1}$.
We thus estimate $\Tr\left(\prod_{i=1}^{M} X_i \tilde \rho\right)/\Tr(\tilde \rho)$ where $\tilde\rho\coloneqq \left(\prod_{i=1}^{M-1}\frac{\mbI + Z_iZ_{i+1}}{2}\right) \mD_p^{\otimes M}(\GHZ(\theta_2)) \left(\prod_{i=1}^{M-1}\frac{\mbI + Z_iZ_{i+1}}{2}\right)$ is the noisy state projected onto the stabilizer subspace. 
To run probabilistic cancellation~\cite{Temme2017error}, we utilize the decomposition of the inverse map for depolarizing noise $\mD_p^{-1}= \left[1+\frac{3p}{4(1-p)}\right]\id -\frac{p}{4(1-p)}(\mX+\mY+\mZ)$ where $\mX,\mY,\mZ$ are unitary Pauli channels.
After every depolarizing channel, we simulate the action of this inverse channel by following the procedure described in Appendix~\ref{app:optimality PEC}.

Following the analysis described above, we plot in Fig.~\ref{fig:accuracy_prob} our lower bound and the actual sample costs to achieve accuracy $\delta=0.1$ and failure probability $\varepsilon=0.2$ for 7-qubit local depolarizing noise with respect to the noise strength $p$. 
We observe that Theorem~\ref{thm:sampling bound} provides nontrivial lower bounds that are comparable to the actual sample cost, the difference being the factor of 3 to 6 in the studied range.  

Note that Fig.~\ref{fig:accuracy_prob} also includes the sample number required for the case when we do not apply any error mitigation but just measure $\prod_{i=1}^M X_i$ on $\mD_p^{\otimes M}(\GHZ(\theta_2))$.
When the noise strength is small, the target accuracy and success probability can be achieved without error mitigation, making the no-mitigation strategy sample-efficient. 
However, as the noise gets stronger, it gets harder to achieve the target performance and eventually becomes impossible to realize the target accuracy no matter what sample number is allowed, as can be seen in the divergent behavior in Fig.~\ref{fig:accuracy_prob}.

\begin{figure}
    \centering
    \includegraphics[width=0.7\columnwidth]{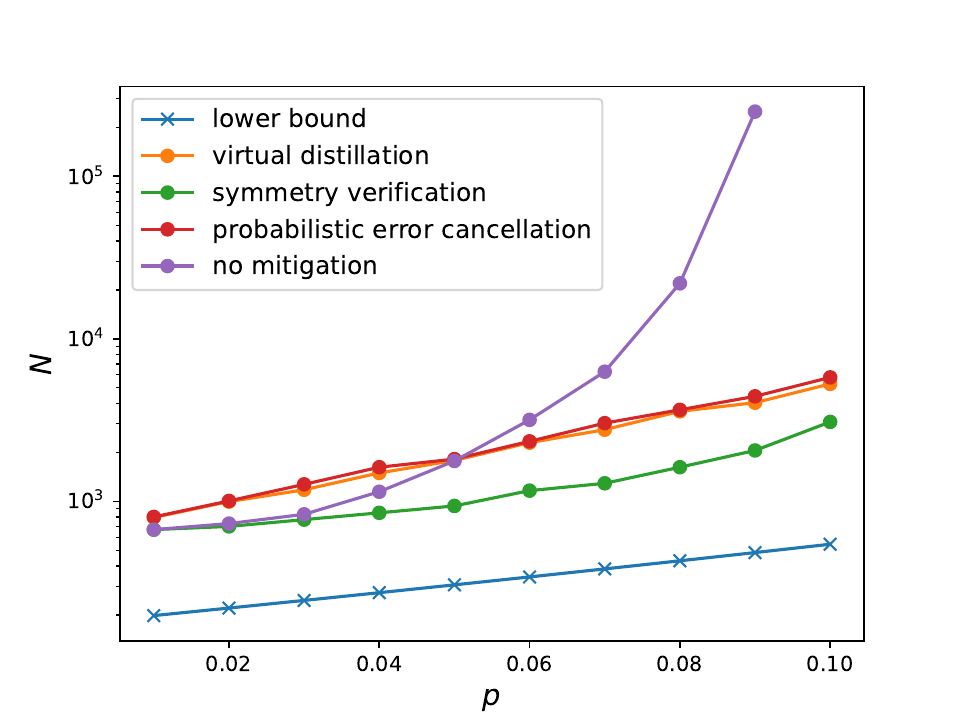}
    \caption{The lower bound in Theorem~\ref{thm:sampling bound} and the actual samples $N$ used for several specific error-mitigation protocols to mitigate 7-qubit local depolarizing noise with noise strength $p$. We fixed accuracy $\delta =0.1$, failure probability $\varepsilon=0.2$ and considered $\mbO=\{\otimes_{i=1}^m X_i\}$ and $\mbS=\{\GHZ(\pm\theta_1), \GHZ(\pm\theta_2)\}$. }
    \label{fig:accuracy_prob}
\end{figure}


\section{Scaling of the bound in Theorem~\ref{thm:sample bound standard deviation} with noise strength}\label{app:optimality PEC}

Let us consider $M$-qubit state suffering from the $M$-qubit local dephasing noise $\mZ_p^{\otimes M}$ where $\mZ(\rho)\coloneqq (1-p)\rho + pZ\rho Z$. 
Consider the class of error mitigation with the same effective noise channels $\mE_i = \mZ_p^{\otimes M},\ \forall i$.
We also have a mild condition that the set $\mbO$ of target observables consists of normalized observables, i.e., those with the unit operator norm, and contains the Pauli operator $\prod_{i=1}^M X_i$, and the set $\mbS$ of target states contains $M$-qubit GHZ states $\GHZ_\pm\coloneqq\dm{\GHZ_\pm}$ with $\ket{\GHZ_\pm}\coloneqq\frac{1}{\sqrt{2}}(\ket{0}^{\otimes M}\pm\ket{1}^{\otimes M})$.  
We also assume that $b_{\max}\leq 1$, which excludes the trivial strategy that always outputs a fixed estimate. 

Noting that $D_\mbO(\GHZ_+,\GHZ_-)=\Tr[\prod_{i=1}X_i (\GHZ_+-\GHZ_-)]=2$, the choice of $\rho=\GHZ_+$ and $\sigma=\GHZ_-$ satisfies $D_\mbO(\rho,\sigma)-2b_{\max}\geq 0$. 
This puts a further bound for the lower bound in \eqref{eq:bias standard deviation bound} and results in 
\bal
 N\geq \frac{\log\left[1-\frac{1}{1+\frac{2\sigma_{\max}^{\rm QEM}}{D_\mbO(\GHZ_+,\GHZ_-)-2b_{\max}}}\right]}{\log[1/F(\mZ_p^{\otimes M}(\GHZ_+),\mZ_p^{\otimes M}(\GHZ_-))]}.
 \label{eq:general bias std local dephasing}
\eal

Note that 
\bal
 \mZ_p^{\otimes M}(\GHZ_\pm) &= \sum_{k=0}^{\lfloor M/2 \rfloor} \binom{M}{2k}(1-p)^{M-2k}p^{2k} \GHZ_\pm\\
 &\quad+\sum_{k=1}^{\lfloor (M+1)/2 \rfloor} \binom{M}{2k-1}(1-p)^{M-2k+1}p^{2k-1} \GHZ_\mp\\
 &=\frac{1+(1-2p)^M}{2}\GHZ_\pm + \frac{1-(1-2p)^M}{2}\GHZ_\mp,
 \label{eq:GHZ local dephasing}
\eal 
where we used 
\bal
\sum_{k=0}^{\lfloor M/2 \rfloor} \binom{M}{2k}(1-p)^{M-2k}p^{2k}&=\sum_{k=0}^{M} \binom{M}{k}(1-p)^{M-k}\frac{p^k+(-p)^{k}}{2}\\
&=\frac{1+(1-2p)^M}{2}.
\eal
\bal
\sum_{k=1}^{\lfloor (M+1)/2 \rfloor} \binom{M}{2k-1}(1-p)^{M-2k+1}p^{2k-1}&=\sum_{k=0}^{M} \binom{M}{k}(1-p)^{M-k}\frac{p^k-(-p)^{k}}{2}\\
&=\frac{1-(1-2p)^M}{2}.
\eal
This gives 
\bal
 F(\mZ_p^{\otimes M}(\GHZ_+),\mZ_p^{\otimes M}(\GHZ_-)) = 1-(1-2p)^{2M}.
\eal

Focusing on the case when $p=\frac{1}{2}-\delta$ with $\delta\ll 1$, the scaling of the lower bound in \eqref{eq:general bias std local dephasing} with respect to the noise strength is given by 
\bal
\left[\log\frac{1}{1-(1-2p)^{2M}}\right]^{-1}\sim 1/(2\delta)^{2M}.
\label{eq:local dephasing error scaling}
\eal

We now show that this scaling can be achieved by probabilistic error cancellation.
The probabilistic error cancellation method is the strategy to mitigate the effect of noise channel $\mE$ by simulating the action of the inverse noise channel $\mE^{-1}$. 
The simulation can be accomplished by considering a linear decomposition of the inverse channel $\mE^{-1}=\sum_i c_i \mB_i$ where $c_i$ is a (possibly negative) real number and $\mB_i$ is a physical operation that can be implemented on the given device.
One then applies the operation $\mB_i$ with probability $|c_i|/\gamma$ with $\gamma\coloneqq \sum_i |c_i|$, makes a measurement, and multiplies $\gamma {\rm sgn}(c_i)$ to the outcome, where ${\rm sgn}(c_i)$ is the sign of $c_i$ that takes +1 if $c_i\geq 0$ and $-1$ if $c_i<0$.
This constructs a distribution whose expected value coincides with the ideal expectation value with the standard deviation scaled by the factor of $\gamma$. 
Therefore, by taking the average over $N$ samples from this distribution, one can construct an estimate of the ideal expectation value, which obeys the distribution with the standard deviation $\sim \gamma/\sqrt{N}$.  

In general, one can also consider optimizing $\gamma$ over the choices of implementable operations $\{\mB_i\}_i$.
Such optimal cost for mitigating the local dephasing noise is characterized as~\cite{Takagi2021fundamental}
\bal
\gamma_{\rm opt} = \frac{1}{(1-2p)^M}.
\eal
This implies that the sample number $N$ that achieves some fixed standard deviation becomes $N\propto \frac{1}{(1-2p)^{2M}}$.
When $p=\frac{1}{2}-\delta$ with $\delta \ll 1$, the sample scales as $N\sim 1/(2\delta)^{2M}$, realizing the same scaling in \eqref{eq:local dephasing error scaling}.


\section{Proof of Theorem~\ref{thm:layer}} \label{app:layer proof}

\begin{proof}

Let $\mD_{p}(\rho)\coloneqq (1-p)\rho + p \mbI/2$ be the single-qubit depolarizing noise. 
Let also $\mU_1,\dots, \mU_L$ be the unitary channels followed by a local depolarizing channel with the noise strength at least $\gamma$ and suppose the $l^{\mathrm{th}}$ layer $\mU_l$ in the $n^{\mathrm{th}}$ noisy circuit is followed by the local depolarizing channel $\otimes_{m=1}^M\mD_{\gamma_{n,l,m}}$. By assumption, we have $\gamma_{n,l,m}\geq\gamma\ \forall n,l,m$.
Also, let $\Lambda_{n,l}$ and $\Xi_{n,l}$ be the unital channels applied before and after the $l^{\rm th}$ layer $U_l$ in the $n^{\rm th}$ noisy layered circuit. 

We first note that every $\mD_{\gamma_{n,l,m}}$ can be written as $\mD_{\gamma_{n,l,m}-\gamma}\circ\mD_{\gamma}$. 
Since the second local depolarizing channel $\otimes_{m=1}^M\mD_{\gamma_{n,l,m}-\gamma}$ is unital, it can be absorbed in the following unital channel $\Xi_{n,l}$, allowing us to focus on $\mD_{\gamma}^{\otimes M}$ as the noise after each layer.

Then, the effective noise channel for the $n^{\rm th}$ layered circuit can be written as 
\bal
 \mE_n &= \mD_{\gamma}^{\otimes M}\circ\Xi_{n,L}\circ\mU_{L}\circ\Lambda_{n,L}\circ\cdots\circ\mD_{\gamma}^{\otimes M} \circ\Xi_{n,1}\circ\mU_{1}\circ\Lambda_{n,1}\circ\mU_1^\dagger\circ\dots\circ\mU_L^\dagger
 \label{eq:effective noise channel layered}
\eal
With a little abuse of notation, we write this as $\mE_n=\prod_{l=1}^{L}\left[\mD_{\gamma}^{\otimes M}\circ\Xi_{n,l}\circ\mU_{l}\circ\Lambda_{n,l}\right]\circ\prod_{l=1}^L \mU_{L-l+1}^\dagger$ where we let the product sign refer to the concatenation of quantum channels.

Let $\rho_{\rm in}$ and $\sigma_{\rm in}$ be some input states and  $\rho=\mU_L\circ\dots\circ\mU_1(\rho_{\rm in})$ and $\sigma=\mU_L\circ\dots\circ\mU_1(\sigma_{\rm in})$. 
Then, the unitary invariance and the triangle inequality of the trace distance imply that 
\bal
D_{\tr}\left(\otimes_{n=1}^N \mE_n(\rho), \otimes_{n=1}^N\mE_n(\sigma)\right)&=D_{\tr}\left(\otimes_{n=1}^{N}\mF_n(\rho_{\rm in}),\otimes_{n=1}^{N}\mF_n(\sigma_{\rm in})\right)\\ 
&\leq  D_{\rm tr}\left(\otimes_{n=1}^{N}\mF_n(\rho_{\rm in}),\frac{\mbI}{2^{N M}}\right) + D_{\rm tr}\left(\otimes_{n=1}^{N}\mF_n(\sigma_{\rm in}),\frac{\mbI}{2^{N M}}\right)
\label{eq:layer trace distance bound}
\eal
where we defined $\mF_n\coloneqq \prod_{l=1}^{L}\left[\mD_{\gamma}^{\otimes M}\circ\Xi_{n,l}\circ\mU_{l}\circ\Lambda_{n,l}\right]$.
Using the quantum Pinsker's inequality 
\bal
 D_{\tr}(\rho,\sigma)\leq \sqrt{\frac{\ln 2}{2}}\,\sqrt{S(\rho\|\sigma)}
\eal
that holds all states $\rho$, $\sigma$, where $S(\rho\|\sigma)\coloneqq \Tr(\rho\log \rho)-\Tr(\rho\log\sigma)$ is the relative entropy, we get 

\begin{equation}\begin{aligned}
 D_{\rm tr}\left(\otimes_{n=1}^{N}\mF_n(\rho_{\rm in}),\frac{\mbI}{2^{N M}}\right)&\leq \sqrt{\frac{\ln 2}{2}} \sqrt{ S\left(\otimes_{n=1}^{N}\mF_n(\rho_{\rm in})\,\Big\|\,\frac{\mbI}{2^{N M}}\right)}\\
 &=\sqrt{\frac{\ln 2}{2}} \sqrt{ \sum_{n=1}^NS\left(\mF_n(\rho_{\rm in})\,\Big\|\,\frac{\mbI}{2^{M}}\right)}
 \label{eq:bound after Pinsker}
\end{aligned}\end{equation}
where in the second line we used the additivity of the relative entropy $S(\rho_1\otimes\rho_2\|\sigma_1\otimes\sigma_2)=S(\rho_1\|\sigma_1)+S(\rho_2\|\sigma_2)$.

We now recall the result in Ref.~\cite{Kastoryano2013quantum}, showing that  
\bal
 S\left(\mD_\gamma^{\otimes M}(\tau)\,\Big\|\,\mbI/2^M\right)\leq (1-\gamma)^2 S(\tau\,\|\,\mbI/2^M) 
 \label{eq:relative entropy Sobolev stochastic Pauli}
\eal
for arbitrary $M$-qubit state $\tau$.
This implies that for arbitrary state $M$-qubit state $\tau$, noise strength $\gamma$, unitary $\mU$, and unital channels $\Xi$, $\Lambda$, 
\bal
 S(\mD_\gamma^{\otimes M}\circ\Xi\circ \mU \circ \Lambda(\tau)\,\|\,\mbI/2^M)&\leq (1-\gamma)^2 S(\Xi\circ \mU \circ \Lambda(\tau)\,\|\,\mbI/2^M)\\
 &= (1-\gamma)^2 S(\Xi\circ \mU \circ \Lambda(\tau)\,\|\,\Xi\circ \mU \circ \Lambda(\mbI/2^M))\\
 &\leq (1-\gamma)^2 S(\tau\,\|\,\mbI/2^M)
\label{eq:relative entropy decrease one layer}
\eal
where we used \eqref{eq:relative entropy Sobolev stochastic Pauli} in the first line, the fact that $\Xi\circ\mU\circ\Lambda(\mbI/2^M)=\mbI/2^M$ because $\Xi$ and $\Lambda$ are unital in the second line, and the data-processing inequality for the relative entropy in the third line. 
The sequential application of \eqref{eq:relative entropy decrease one layer} results in 
\bal
S\left(\mF_n(\rho_{\rm in})\,\Big\|\,\frac{\mbI}{2^{M}}\right)
&\leq (1-\gamma)^{2L}\,S\left(\rho_{\rm in}\,\Big\|\,\frac{\mbI}{2^{M}}\right)\\
 &\leq (1-\gamma)^{2L} M,
 \label{eq:bound on relative entropy}
\eal
where the second line follows from the upper bound of the relative entropy. 
Combining \eqref{eq:layer trace distance bound}, \eqref{eq:bound on relative entropy} and \eqref{eq:bound after Pinsker}, we get 
\bal
D_{\rm tr}\left(\otimes_{n=1}^{N}\mF_n(\rho_{\rm in}),\otimes_{n=1}^{N}\mF_n(\sigma_{\rm in})\right)\leq 2\sqrt{\frac{\ln 2}{2}}\sqrt{(1-\gamma)^{2L}M N}. 
\label{eq:trace distance exponential upper bound}
\eal
Combining this with \eqref{eq:trace distance bound}, we get 
\bal
 N\geq \frac{(1-2\varepsilon)^2}{2\ln(2)M(1-\gamma)^{2L}}
\eal
for any $\varepsilon\leq 1/2$. 
\end{proof}


\section{Bound for layered circuits with a fixed target bias and standard deviation}\label{app:layer standard deviation proof}

Theorem~\ref{thm:layer} in the main text shows the exponential sampling overhead for noisy layered circuits with a fixed target accuracy and success probability. 
Here, we present a similar exponential growth of the required sample overhead for a fixed target bias and standard deviation.

\begin{thm} \label{thm:layer standard deviation}
Suppose that an error-mitigation strategy described above is applied to an $M$-qubit circuit to mitigate local depolarizing channels with strength at least $\gamma$ that follow $L$ layers of unitaries. Then, if the estimator of the error mitigation has the maximum standard deviation $\sigma_{\max}^{\QEM}$ and maximum bias $b_{\max}$, the required number $N$ of samples is lower bounded as
\bal
 N\geq \frac{1}{4 M(1-\gamma)^{2L}}\left(\frac{1}{\frac{2\sigma_{\max}^{\QEM}}{D_\mbO(\rho,\sigma)-2b_{\max}}+1}\right)^2
\eal
for arbitrary states $\rho,\sigma\in\mbS$ such that $D_\mbO(\rho,\sigma)-2b_{\max}\geq 0$.
\end{thm}

\begin{proof}
We first note that $D_F(\rho,\sigma)\leq \sqrt{S(\rho\|\sigma)}$ for arbitrary states $\rho,\sigma$~\cite{Hayashi2016quantum}.
Also, the purified distance $D_F$ satisfies the triangle inequality~\cite{Gilchrist2005distance}. 
Therefore, we can repeat the argument in \eqref{eq:layer trace distance bound} and \eqref{eq:bound after Pinsker} with the purified distance to get 
\bal
 D_F(\otimes_{n=1}^N\mE_n(\rho),\otimes_{n=1}^N\mE_n(\sigma))\leq \left(\sqrt{\sum_{n=1}^N S\left(\mF_n(\rho_{\rm in})\,\Big\|\,\frac{\mbI}{2^{M}}\right)}+\sqrt{\sum_{n=1}^N S\left(\mF_n(\sigma_{\rm in})\,\Big\|\,\frac{\mbI}{2^{M}}\right)}\right).
\eal
Using \eqref{eq:bound on relative entropy}, we further have 
\bal
 D_F(\otimes_{n=1}^N\mE_n(\rho),\otimes_{n=1}^N\mE_n(\sigma))\leq 2\sqrt{M N(1-\gamma)^{2L}}
 \label{eq:purified distance upper bound exponential}
\eal
This allows us to put another bound to \eqref{eq:standard deviation lower bound observable independent} to get 
\bal
\sigma_{\max}^{\QEM}\geq\max_{\substack{\rho,\sigma\in\mbS\\D_\mbO(\rho,\sigma)-2b_{\max}\geq 0}} \frac{1}{2}\left(\frac{1}{2\sqrt{M N(1-\gamma)^{2L}}}-1\right)(D_\mbO(\rho,\sigma)-2b_{\max}).
\eal

Reorganizing this, we obtain
\bal
 N\geq \frac{1}{4 M(1-\gamma)^{2L}}\left(\frac{1}{\frac{2\sigma_{\max}^{\QEM}}{D_\mbO(\rho,\sigma)-2b_{\max}}+1}\right)^2.
\eal
for arbitrary state $\rho,\sigma$ such that $D_\mbO(\rho,\sigma)-2b_{\max}\geq 0$.
\end{proof}


\section{Alternative bounds for noisy layered circuits}\label{app:layered alterative}

Let us investigate alternative bounds that also diverge with the vanishing failure probability $\varepsilon\to 0$. 
We first note the inequality~\cite{Hayashi2016quantum} 
\bal
 S(\tau\|\eta)\geq \log F(\tau,\eta)^{-1}
 \label{eq:upper bound for log fidelity with relative entropy}
\eal
that hold for arbitrary states $\tau$ and $\eta$. 
Together with Theorem~\ref{thm:sampling bound}, our goal is to upper bound $S(\mE_n(\rho)\|\mE_n(\sigma))$ for the effective noise channel defined in \eqref{eq:effective noise channel layered}. 

To do this, we employ the continuity bounds of the relative entropy~\cite{Audenaert2005continuity}. Namely, for arbitrary states $\tau$ and a full-rank state $\eta$,
\bal
 S(\tau\|\eta)&\leq \frac{4D_{\tr}(\tau,\eta)^2}{\lambda_{\min}(\eta)}
 \label{eq:upper bound relative entropy trace squared}
\eal
where $\lambda_{\min}(\eta)$ is the minimum eigenvalue of a state $\eta$, and for $D_{\tr}(\tau,\eta)\leq 1/2$, 
\bal
 S(\tau\|\eta)&\leq \frac{2}{\ln 2}D_{\tr}(\tau,\eta)\left[\log d + \log\left(\frac{2}{D_{\rm tr}(\tau,\eta)}\right) + \frac{1}{2}\log(1/\lambda_{\min}(\eta))\right].
 \label{eq:upper bound relative entropy lambda min log}
\eal
where $d$ is the dimension of the Hilbert space that $\tau$ and $\eta$ act on. 
Note that \eqref{eq:upper bound relative entropy trace squared} has the square dependence on the trace distance, while \eqref{eq:upper bound relative entropy lambda min log} has the log dependence on the minimum eigenvalue.

Note that $\mD_{\gamma}^{\otimes M}=(1-(4\gamma)^M)\Psi + (4\gamma)^M\mD^{\otimes M}$ where $\Psi$ is some quantum channel and $\mD$ is the completely depolarizing channel. 
This ensures that 
\bal
\lambda_{\min}(\mD_{\gamma}^{\otimes M}(\sigma))\geq \left(\frac{\gamma}{2}\right)^M
\label{eq:lower bound for minimum eigenvalue one layer}
\eal
for arbitrary state $\sigma$.
Noting the structure in \eqref{eq:effective noise channel layered}, where all components preserve the maximally mixed state, this results in 
\bal
\lambda_{\min}(\mE_{n}(\sigma))\geq \left(\frac{\gamma}{2}\right)^M.
\label{eq:lower bound for minimum eigenvalue}
\eal
We combine \eqref{eq:upper bound relative entropy trace squared} with \eqref{eq:trace distance exponential upper bound} and \eqref{eq:lower bound for minimum eigenvalue} to get 
\bal
 S(\mE_n(\rho)\|\mE_n(\sigma)))\leq 8\ln (2)\,M(1-\gamma)^{2L}\left(\frac{2}{\gamma}\right)^M.
 \label{eq:relative entropy upper bound exponential with qubit}
\eal

Combining this with \eqref{eq:upper bound for log fidelity with relative entropy} and Theorem~\ref{thm:sampling bound}, we obtain 
\bal
 N &\geq \log\left[\frac{1}{4\varepsilon(1-\varepsilon)}\right]\frac{\left(\frac{\gamma}{2}\right)^M}{8\ln (2)\,M(1-\gamma)^{2L}}.
 \label{eq:sample lower bound fix probability exponential with qubit number}
\eal
This bound diverges with $\varepsilon\to 0$ and grows exponentially with the circuit depth with exponent $2L$. However, this is also exponentially loose with the qubit number $M$.

The second bound \eqref{eq:upper bound relative entropy lambda min log} can remedy this drawback at the cost of a looser circuit depth dependence. 
Let us first note that the function $x\log(1/x)$ is monotonically increasing for $0<x\leq 1/2$.
Therefore, for $M$, $\gamma$, and $L$ such that $D_{\tr}(\mE_n(\rho),\mE_n(\sigma))\leq \sqrt{2\ln 2}\sqrt{M}(1-\gamma)^L\leq 1/2$, we can combine \eqref{eq:upper bound relative entropy lambda min log}, \eqref{eq:trace distance exponential upper bound}, and \eqref{eq:lower bound for minimum eigenvalue} to get 
\bal
 S(\mE_n(\rho)\|\mE_n(\sigma))\leq 2\sqrt{\frac{2}{\ln 2}}\sqrt{M}(1-\gamma)^L\left[M+\log\left(\frac{2}{\sqrt{2\ln 2}\sqrt{M}(1-\gamma)^L}\right)+\frac{M}{2}\log(2/\gamma)\right].
\label{eq:relative entropy upper bound poly qubit}
\eal

This, together with Theorem~\ref{thm:sampling bound}, we get 
\bal
 N\geq \log\left[\frac{1}{4\varepsilon(1-\varepsilon)}\right]\frac{\sqrt{\ln 2}}{\sqrt{8 M}(1-\gamma)^L}\left[L\log\left(\frac{2}{1-\gamma}\right)+M-\frac{1}{2}\log\left(\frac{M\ln 2}{2}\right)+\frac{M}{2}\log(2/\gamma)\right]^{-1}.
\eal
This has a polynomial dependence with $M$ and still grows exponentially with $L$ (up to a polynomial correction), while the $L$ dependence is not as good as \eqref{eq:sample lower bound fix probability exponential with qubit number}. 

We can also apply a similar argument to show bounds for fixed standard deviation and bias, which diverge with $\sigma_{\max}^{\QEM}\to 0$.
Namely, we combine \eqref{eq:relative entropy upper bound exponential with qubit} with \eqref{eq:upper bound for log fidelity with relative entropy} and Theorem~\ref{thm:sample bound standard deviation} to obtain 
\bal
  N \geq \log\left[1-\frac{1}{\left(1+\frac{2\sigma_{\max}^{\QEM}}{D_\mbO(\rho,\sigma)-2b_{\max}}\right)^2
}\right]^{-1}\frac{\left(\frac{\gamma}{2}\right)^M}{8\ln (2)\,M(1-\gamma)^{2L}}.
\label{eq:sample lower bound fix standard deviationexponential with qubit number}
\eal
This bound diverges with $\sigma_{\max}^{\QEM}\to 0$ and grows exponentially with the circuit depth with exponent $2L$.  
Similarly, using \eqref{eq:relative entropy upper bound poly qubit} and Theorem~\ref{thm:sample bound standard deviation}, we get 
\bal
 N\geq \log\left[1-\frac{1}{\left(1+\frac{2\sigma_{\max}^{\QEM}}{D_\mbO(\rho,\sigma)-2b_{\max}}\right)^2
}\right]^{-1}\frac{\sqrt{\ln 2}}{\sqrt{8 M}(1-\gamma)^L}\left[L\log\left(\frac{2}{1-\gamma}\right)+M-\frac{1}{2}\log\left(\frac{M\ln 2}{2}\right)+\frac{M}{2}\log\left(\frac{2}{\gamma}\right)\right]^{-1}.
\eal


\section{Noisy layered circuits under other noise models}\label{app:layered general noise}

Let us consider a layered circuit affected by a noise channel $\mN_{n,l}$ after the $l^{\mathrm{th}}$ layer $U_l$ in the $n^{\mathrm{th}}$ copy of the noisy circuits. 
Suppose that for every $n$, all $\{\mN_{n,l}\}_{l=1}^L$ and $\{U_l\}_{l=1}^L$ share the same fixed point $\sigma_n$.
We aim to mitigate errors by applying additional channels $\Lambda_{n,l}$ and $\Xi_{n,l}$ that also preserve $\sigma_n$ before and after $U_l$ in the $n^{\mathrm{th}}$ noisy circuit and applying a global trailing quantum process over $n$ copies of distorted state. 

For an arbitrary quantum channel $ \mN$ with a fixed point $\sigma$, let $\xi_\mN$ be a contraction constant for $\mN$ that satisfies
\bal
 S(\mN(\rho)\|\sigma) \leq \xi_\mN S(\rho\|\sigma),\ \forall \rho.
 \label{eq:contraction relative entropy}
\eal
The contraction constant $\xi_\mN$ was studied for various situations, particularly in relation to logarithmic Sobolev inequalities for dynamical semigroups that represent continuous Markovian evolution~\cite{Kastoryano2013quantum,Carbone2015logarithmic,Kastoryano2016quantum,MullerHermes2016relative,MullerHermes2016entropy,Bardet2021onthemodified,Beigi2020quantum,Capel2020themodified,Hirche2020oncontraction}. 

Using the contraction constant, we can state the sample lower bound in the general form, which reproduces  Theorems~\ref{thm:layer}~and~\ref{thm:layer standard deviation} as special cases.

\begin{thm}\label{thm:layer with contraction factor}
Suppose that an error-mitigation strategy achieves \eqref{eq:accuracy probability requirement} with some $\delta\geq 0$ and $0\leq \varepsilon\leq 1/2$ to mitigate noise $\{\mN_{n,l}\}_{n,l}$ that occurs after each $U_l$ in the $n^{\mathrm{th}}$ layered circuit in an error mitigation strategy described above. 
Suppose also that every noise channel has the contraction constant smaller than $\xi$, i.e., $\xi_{\mN_{n,l}}\leq \xi, \forall n,l$. 

Then, if there exist at least two states $\rho,\sigma\in\mbS$ such that $D_\mbO(\rho,\sigma)\geq 2\delta$, the required sample number $N$ is lower bounded as
\bal
 N\geq \frac{(1-2\varepsilon)^2}{2\ln (2)\,M \xi^{L}}.
\eal

Similarly, to achieve the maximum standard deviation $\sigma_{\max}^{\QEM}$ and maximum bias $b_{\max}$, the required number $N$ of samples is lower bounded as
\bal
 N\geq \frac{1}{4 M \xi^{L}}\left(\frac{1}{\frac{2\sigma_{\max}^{\QEM}}{D_\mbO(\rho,\sigma)-2b_{\max}}+1}\right)^2
\eal
for arbitrary states $\rho,\sigma\in\mbS$ such that $D_\mbO(\rho,\sigma)-2b_{\max}\geq 0$.
\end{thm}

\begin{proof}

Both statements can be obtained essentially by following the same arguments in the proofs of Theorems~\ref{thm:layer}~and~\ref{thm:layer standard deviation}.
Namely, defining $\mF_n\coloneqq \prod_{l=1}^{L}\left[\mN_{n,l}\circ\Xi_{n,l}\circ\mU_{l}\circ\Lambda_{n,l}\right]$, the argument leading to \eqref{eq:trace distance exponential upper bound} gives us
\bal
 D_{\rm tr}\left(\otimes_{n=1}^N \mF_n(\rho_{\rm in}), \otimes_{n=1}^N \mF_n(\sigma_{\rm in})\right)\leq \sqrt{2\ln2}\sqrt{\xi^L M N}.
\eal
This, together with \eqref{eq:trace distance bound}, gives 
\bal
 N\geq \frac{(1-2\varepsilon)^2}{2\ln(2)M\xi^L}
\eal
whenever there exist two states $\rho,\sigma\in\mbS$ such that $D_\mbO(\rho,\sigma)\geq 2\delta$. 

Similarly, the argument to get \eqref{eq:purified distance upper bound exponential} results in 
\bal
 D_F\left(\otimes_{n=1}^N \mF_n(\rho_{\rm in})\|\otimes_{n=1}^N \mF_n(\sigma_{\rm in})\right)\leq 2\sqrt{M N\xi^L}.
\eal

Combining this to \eqref{eq:standard deviation lower bound observable independent} leads to
\bal
 N\geq \frac{1}{4 M\xi^{L}}\left(\frac{1}{\frac{2\sigma_{\max}^{\QEM}}{D_\mbO(\rho,\sigma)-2b_{\max}}+1}\right)^2.
\eal
for arbitrary state $\rho,\sigma$ such that $D_\mbO(\rho,\sigma)-2b_{\max}\geq 0$.

\end{proof}

Theorems~\ref{thm:layer}~and~\ref{thm:layer standard deviation} are recovered by noting that the contraction constant for the local depolarizing noise $\mD_\gamma^{\otimes M}$ is known as $\xi_{\mD_\gamma^{\otimes M}}=(1-\gamma)^2$~\cite{Kastoryano2013quantum}. 

We remark that analogous results to Theorem~\ref{thm:layer with contraction factor} can also be obtained in terms of the contractivity constant for other distance measures, such as trace distance, purified distance, and R\'enyi-$\alpha$ divergences with $\alpha>1$. 
Indeed, we can get corresponding bounds for stochastic Pauli noise by employing the contractivity for R\'enyi-$2$ sandwiched relative entropy\cite{Hirche2020oncontraction}
\bal
 S_2(\Lambda^{\otimes n}(\rho)\|\mbI/d^n)\leq \exp\left[2(1-1/n)\frac{\ln r}{\ln n}\right] S_2(\rho\|\mbI/d^n)
 \label{eq:contraction 2-Renyi}
\eal
for every $n$ and $d$-dimensional unital channel $\Lambda$ that satisfies $r^{-1}\|\Lambda-(1-r)\mD\|_{2\to 2}\leq 1$.
Here, $S_2(\tau\|\eta)\coloneqq \log \Tr\left(\eta^{-1/2}\,\tau\,\eta^{-1/2}\,\tau\right)$ is the R\'enyi-2 sandwiched relative entropy~\cite{MullerLennert2013onquantum}, $\mD$ is the completely depolarizing channel, and $\|\Phi\|_{2\to2}\coloneqq \sup_{X\neq 0}\frac{\|\Phi(X)\|_2}{\|X\|_2}$ is the $2\to 2$ norm for an arbitrary linear map $\Phi$.
For an error probability ${\bf q}=(q_x,q_y,q_z)$ for each Pauli error, define a qubit stochastic Pauli noise as $\mT_{\bf q}(\tau)\coloneqq (1-q_x-q_y-q_z)\tau + q_x X\tau X + q_y Y\tau Y + q_z Z\tau Z$. 
Then, by taking $r=q\coloneqq|1-2\min\{q_x+q_y,q_y+q_z,q_x+q_z\}|$ in Eq.~\eqref{eq:contraction 2-Renyi}, we obtain 
\bal
 S_2(\mT_{\bf q}^{\otimes n}(\rho)\|\mbI/2^n)\leq q^{1/\ln2}S_2(\rho\|\mbI/2^n).
\eal
We remark that this type of bound was also employed in Ref.~\cite{Wang2021noise-induced}.
Since $S(\tau\|\eta)\leq S_2(\tau\|\eta),\ \forall \tau,\eta$~\cite{MullerLennert2013onquantum}, we can upper bound the form in \eqref{eq:bound after Pinsker} by \eqref{eq:contraction 2-Renyi}.
From there, we can follow the same argument as the ones in the proofs of Theorems~\ref{thm:layer}~and~\ref{thm:layer standard deviation}, noting the data-processing inequality and the additivity for tensor-product states known for R\'eny-2 sandwiched relative entropy~\cite{MullerLennert2013onquantum}. 
This results in the sampling lower bounds $\Omega\left(q^{L/\ln2}\right)$ for stochastic Pauli noise. 

More generally, the exponential sample growth can be established for a large class of tensor-product channels. 
Let $\{\Phi_t\}_{t\geq 0}$ be a set of unital qubit channels constructing a dynamical semigroup, i.e., $\Phi_0 = \id$ and $\Phi_{s+r} = \Phi_s\circ\Phi_r, \forall s,t\geq 0$, generated by a Liouvillian $\mL$, i.e., $\Phi_t=e^{t\mL}$~\cite{Gorini1976completely}. 
Suppose further that $\mbI/2$ is the unique solution of $\mL(\tau)=0$, i.e., the unique fixed point for $\{\Phi_t\}_t$.
Then, for a noisy channel $\mN=\Phi_t^{\otimes M}$ defined for some $t>0$ comes with $\xi_\mN<1$~\cite{MullerHermes2016entropy}, leading to the exponential sample cost together with Theorem~\ref{thm:layer with contraction factor}. 
An example includes a channel generated by the Liouvillian constructed by a stochastic Pauli channel $\mL(\tau)=\sum_{i=0}^3 q_iP_i\tau P_i  - \tau$ where $q_i$ is the probability that $i^{\mathrm{th}}$ Pauli $P_i$ is applied.

We can also apply this to noise channels globally applied to $M$ qubits. 
Consider the $M$-qubit (not necessarily unital) depolarizing channel with a full-rank fixed point $\sigma$ defined as $\mD_{\gamma,\sigma}(\tau)\coloneqq (1-\gamma)\tau + \gamma \sigma$. 
For instance, the fixed state $\sigma$ can be taken as the thermal Gibbs state $\sigma_{\beta,H}=e^{-\beta H}/\Tr(e^{-\beta H})$ for some inverse temperature $\beta$ and Hamiltonian $H$. 
Then, it was shown in Ref.~\cite{MullerHermes2016relative} that 
\bal
 S(\mD_{\gamma,\sigma}(\rho)\|\sigma)\leq (1-\gamma)^{2\alpha_1} S(\rho\|\sigma)
\eal
where $\alpha_1$ is computed by 
\bal
 \alpha_1 = \min_{x\in[0,1]}\frac{1}{2}(1+q_{\lambda_{\min}(\sigma)}(x)).
\eal
Here, $\lambda_{\min}(\sigma)$ is the minimum eigenvalue of $\sigma$, and $q_x(y)$ is defined for $x,y\in(0,1)$ as 
\bal
 q_y(x)\coloneqq\begin{cases}
 \frac{D_2(y\|x)}{D_2(x\|y)} & x\neq y\\
 1 & x=y
 \end{cases}
\eal
where $D_2(x\|y)\coloneqq x\log(x/y)+(1-x)\log[(1-x)/(1-y)]$ is the binary relative entropy. 
As shown in Ref.~\cite{MullerHermes2016relative}, $\alpha_1$ monotonically increases with $\lambda_{\min}(\sigma)$ and behaves as $\alpha_1\to 1/2$ in the limit $\lambda_{\min}(\sigma)\to 0$ and $\alpha_1\to 1$ in the limit $\lambda_{\min}(\sigma)\to 1/2$. 
Taking $\xi_{\mD_{\gamma,\sigma}}=(1-\gamma)^{2\alpha_1}$ in Theorem~\ref{thm:layer with contraction factor} then allows us to obtain sample lower bounds for the generalized global depolarizing noise with an arbitrary full-rank fixed state. 

Similarly to the tensor-product channels, the exponential sample cost can be extended to a wide class of quantum channels. 
Let $\mN$ be a unital channel and $\mN^\dagger$ be its dual, i.e., $\Tr(X\mN(Y))=\Tr(\mN^\dagger(Y)X), \forall X,Y$. 
Then, if $\mN^\dagger\circ\mN$ is primitive, i.e., $\lim_{n\to\infty} (\mN^\dagger\circ\mN)^n(\tau) = \mbI/2^M,\ \forall \tau$, we have $\lambda_\mN<1$~\cite{MullerHermes2016entropy}.

Our results also admit a thermodynamic interpretation.
Consider the situation where the ideal circuit is the identity map and our system suffers from the continuous thermal map $\{\Phi_t\}_t$ generated by a Liouvillian $\mL$ with a thermal Gibbs state $\sigma_{\beta,H}=e^{-\beta H}/\Tr(e^{-\beta H})$ as a fixed state, i.e., $\mL(\sigma_{\beta,H})=0$~\cite{Kastoryano2013quantum}.
This represents the state preservation scenario, in which we would like to extract the classical information about the initial state after some period of time. 

Let $\rho$ be the initial state and $\rho_t=\Phi_t(\rho)$ be a state at time $t\geq 0$. 
Note that for an arbitrary state $\tau$, we have $S(\tau\|\sigma_{\beta,H})=\beta(F(\tau)-F_{\rm eq})$ where $F(\tau)\coloneqq \Tr(\tau H)-S(\tau)/\beta$ is the nonequilibrium free energy of $\tau$ and $F_{\rm eq}\coloneqq F(\sigma_{\beta,H})$ is the equilibrium free energy.
Then, Theorem~\ref{thm:sampling bound} implies that, to extract the expectation value of the initial state $\rho$, we need to use 
\bal
N =\Omega\left(\frac{1}{F(\rho_t)-F_{\rm eq}}\right)
\eal
samples.
In other words, the necessary sampling cost grows as the free energy gets lost due to thermalization.
The property of the dynamical semigroup further gives 
\bal
 F(\rho_t)-F_{\rm eq} \leq e^{-\alpha_{\rm ent} t}\left[F(\rho)-F_{\rm eq}\right]
\eal 
with
\bal
\alpha_{\rm ent} = \min_\tau \frac{\dot{\Sigma}(\tau)}{\beta\left[F(\tau)-F_{\rm eq}\right]}
\eal
where $\dot{\Sigma}(\tau)\coloneqq \left[\frac{d}{dt}S(\tau_t)-\beta \frac{d}{dt}\Tr(\tau_t H)\right]|_{t=0}=-\Tr[\mL(\tau)\ln(\tau)]-\Tr[\mL(\tau)H]$ is the entropy production rate.  
In particular, when  $\alpha_{\rm ent}>0$ holds, i.e., $\dot{\Sigma}(\tau)>0,\ \forall \tau\neq \sigma_{\beta,H}$, this implies
\bal
N = \Omega\left(e^{\alpha_{\rm ent} t}\right),
\eal
the exponential sample cost with time $t$.
This cannot be avoided even if we apply intermediate operations that preserve the Gibbs state, as such operations cannot increase the free energy. 
We also remark that there has been an intense study on the evaluation of the exponent $\alpha_1$ for various types of thermal maps, lattices, and Hamiltonian (e.g., Refs.~\cite{Kastoryano2013quantum,Temme2014hypercontractivity,Temme2015howfast,Capel2018quantum,Capel2020themodified,Beigi2020quantum}). 
These results can directly be applied to give lower bounds for the sampling overhead required for quantum error mitigation.

Thermal channels construct an important class of noise channels of physical significance, and the exponential contraction of the relative entropy, i.e., \eqref{eq:contraction relative entropy} with $\xi_\mN<1$, was established for the thermal map that constructs a Markov semi-group~\cite{Kastoryano2013quantum}. 
It was conjectured in Ref.~\cite{Kastoryano2013quantum} that an arbitrary Markov semi-group with a unique fixed point would also show the exponential decay of the relative entropy.
If this is true, then Theorem~\ref{thm:layer with contraction factor} applies to an even wider class of noise channels, particularly when $U_l=\mbI\,\forall l$ corresponding to the state preservation scenario.

Although this conjecture appears to be plausible, it is still generally elusive to rigorously show the exponential contraction of the relative entropy for non-unital channels. 
Nevertheless, our framework admits an alternative bound that also grows exponentially with circuit depth, as we show in the following.

The idea is to employ the exponential contraction of the trace distance instead of the relative entropy. 
For a channel $\mN$ that has the unique fixed point $\tau$, it holds that~\cite{Temme2010chi}
\bal
 D_{\rm tr}\left(\mN^k(\rho), \tau\right) \leq 2\left[s_2(\mN)\right]^k\left[\lambda_{\min}(\tau)\right]^{-1/2}\quad \forall \rho
 \label{eq:exponential contraction trace distance}
\eal
for an arbitrary positive integer $k$, where $\mN^k$ refers to the $k$ sequential applications of $\mN$, $\lambda_{\min}(\tau)$ is the smallest eigenvalue of $\tau$, and $s_{2}(\mN)$ is the second largest singular value of the map $\Gamma_\tau^{-1/2}\circ \mN \circ \Gamma_\tau^{1/2}$ with $\Gamma_\tau$ defined by $\Gamma_{\tau}(\rho) \coloneqq \tau^{1/2}\rho \tau^{1/2}$.

Consider the setting of the layered circuit with $U_l=\mbI\,\forall l$, where we are to use $N$ $M$-qubit circuits as an input to error mitigation.  
Let $\mN$ be the $M$-qubit noise channel that applies to each layer of a noisy circuit. 
This is the case where the effective noise channel of the depth-$L$ noisy circuit is represented by $\mN^L$.
Then, the bound in Theorem~\ref{thm:sampling bound} in terms of the relative entropy can be bounded as 
\bal
 N\geq \max_{\substack{\rho,\sigma\in\mbS\\D_\mbO(\rho,\sigma)\geq 2\delta}}\frac{2(1-2\varepsilon)^2}{\ln 2 \cdot S(\mN^L(\rho)\|\mN^L(\sigma))} &\geq \max_{\substack{\rho,\sigma\in\mbS\\D_\mbO(\rho,\sigma)\geq 2\delta}}\frac{(1-2\varepsilon)^2\lambda_{\min}(\mN^L(\sigma))}{2\ln 2 \, D_{\rm tr}(\mN^L(\rho),\mN^L(\sigma))^2} \\
 &\geq \frac{(1-2\varepsilon)^2\lambda_{\min}(\mN)\lambda_{\min}(\tau)}{8\ln 2 \left[s_2(\mN)\right]^{2L} }
 \label{eq:sample lower bound non-unital}
\eal
where the second inequality is due to the continuity bound of the relative entropy \eqref{eq:upper bound relative entropy trace squared}, and in the last inequality we used \eqref{eq:exponential contraction trace distance} and defined $\lambda_{\min}(\mN)\coloneqq \min_\rho \lambda_{\min}(\mN(\rho))$, which refers to the minimum eigenvalue of the output state of the channel $\mN$.
This bound shows the exponential sampling cost for an arbitrary channel $\mN$ that has a unique fixed state and $s_2(\mN)<1$.

We remark that $\lambda_{\min}(\mN)$ and $\lambda_{\min}(\tau)$ can be exponentially small with the qubit number $M$, which prevents one from using this bound to argue that the circuit with $L=\Omega(M)$ must use an exponential sample cost, unlike the one in Theorem~\ref{thm:layer with contraction factor}.  
Nevertheless, it still ensures exponential growth with respect to the layer number $L$ if we see it as a separate parameter from the qubit number $M$. 
We also stress that the tighter bound in Theorem~\ref{thm:layer with contraction factor} would become valid once the exponential decay of the relative entropy for non-unital channels is established, going along with future developments in the understanding of non-unital Markovian semi-groups. 

As an example of an important non-unital channel, let us discuss the generalized amplitude damping noise~\cite{Nielsen2011quantum}, which describes the thermal relaxation at finite temperature in, e.g., superconducting systems~\cite{qiskit_noise}. 
The generalized amplitude damping $\mA_{\gamma,\kappa}$ characterized by two parameters $\gamma,\kappa$ with $0< \gamma,\kappa < 1$ has the Kraus operators 
\bal
 K_1 &= \sqrt{1-\kappa}\left(\dm{0}+\sqrt{1-\gamma}\dm{1}\right)\quad
 K_2 = \sqrt{\gamma(1-\kappa)}\ketbra{0}{1}\\
 K_3 &= \sqrt{\kappa}\left(\sqrt{1-\gamma}\dm{0}+\dm{1}\right)\quad
 K_4 = \sqrt{\gamma \kappa}\ketbra{1}{0}.
\eal

The generalized amplitude damping channel maps the Pauli operators as 
\bal
 \mA_{\gamma,\kappa}(X) &= \sqrt{1-\gamma} X\quad \mA_{\gamma,\kappa}(Y) = \sqrt{1-\gamma} Y\\
 \mA_{\gamma,\kappa}(Z) &= (1-\gamma) Z\quad 
 \mA_{\gamma,\kappa}(\mbI) = \mbI + \gamma(1-2\kappa)Z.
\eal
This ensures that at the limit of the infinite number of applications of $\mA_{\gamma,\kappa}$, the three non-trivial Pauli operators $X$, $Y$, and $Z$ vanish, while the identity operator approaches $\mbI + \gamma(1-2\kappa)\left(\sum_{k=0}^\infty (1-\gamma)^k\right) Z = \mbI + (1-2\kappa)Z$.
This shows that the fixed state of $\mA_{\gamma,\kappa}$ is $\frac{\mbI + (1-2\kappa)Z}{2}$, and more generally, the fixed point of the tensor-product channel $\mA_{\gamma,\kappa}^{\otimes k}$ for an arbitrary positive integer $k$ is $\left(\frac{\mbI + (1-2\kappa)Z}{2}\right)^{\otimes k}$, because whatever the $k$-qubit initial state we start from, only the $\mbI/2^k$ term in the Pauli decomposition of the initial state survives after the infinite number of applications of $\mA_{\gamma,\kappa}^{\otimes k}$, which converges to $\left(\frac{\mbI + (1-2\kappa)Z}{2}\right)^{\otimes k}$.
This means that $\mA_{\gamma,\kappa}^{\otimes k}$ has a unique fixed state, to which the bound \eqref{eq:sample lower bound non-unital} applies.

In particular, the sampling lower bound for mitigating the local generalized amplitude damping noise applied to $M$-qubit circuit is obtained by setting $\mN=\mA_{\gamma,\kappa}^{\otimes M}$ in \eqref{eq:sample lower bound non-unital} with $\tau = \left(\frac{\mbI + (1-2\kappa)Z}{2}\right)^{\otimes M}$. 
This immediately gives $\lambda_{\min}(\tau) = \tilde \kappa^M$ where $\tilde \kappa \coloneqq \min\{\kappa, 1-\kappa\}$.

To get $\lambda_{\min}(\mA_{\gamma,\kappa})$, we note that for a general channel $\mN$, if $\mN$ satisfies  $\mN(\rho) \geq  p \mbI$ for every state $\rho$, we have $\lambda_{\min}(\mN)\geq p$. 
This happens when its Choi state $J_\mN=\id\otimes\mN (\tilde \Phi)$ for $\mN$, where $\tilde \Phi=\ketbra{ii}{jj}$ is the unnormalized maximally entangled state,
has the minimum eigenvalue $\lambda_{\min}(J_\mN)\geq p$.
This allows us to lower bound $\lambda_{\min}(\mN)$ by looking at the eigenvalue of the Choi state. 
A straightforward calculation reveals that the Choi state of $\mA_{\gamma,\kappa}$ is given by 
\bal
 J_{\mA_{\gamma,\kappa}}=\begin{pmatrix}
 1-\gamma\kappa & 0 & 0 & \sqrt{1-\gamma}\\
 0 & \gamma \kappa & 0 & 0 \\
 0 & 0 & \gamma(1-\kappa) & 0 \\
 \sqrt{1-\gamma} & 0 & 0 & 1-\gamma(1-\kappa)
 \end{pmatrix},
\eal
which gives $\lambda_{\min}(J_{\mA_{\gamma,\kappa}^{\otimes M}})=\left[\lambda_{\min}(J_{\mA_{\gamma,\kappa}})\right]^M = \left(1-\gamma/2 - \sqrt{1-\gamma + \gamma^2(\kappa - 1/2)^2}\right)^M$.

Finally, $s_2(\mA_{\gamma,\kappa}^{\otimes M})$ can be obtained by computing the singular values of $\Gamma_\tau^{-1/2}\circ\mA_{\gamma, \kappa}^{\otimes M}\circ \Gamma_\tau^{1/2}$. 
Letting $\tau_1 \coloneqq \frac{\mbI+(1-2\kappa)Z}{2}$ so that $\tau=\tau_1^{\otimes M}$, we get $\Gamma_\tau^{-1/2}\circ\mA_{\gamma, \kappa}^{\otimes M}\circ \Gamma_\tau^{1/2}=\left(\Gamma_{\tau_1}^{-1/2}\circ\mA_{\gamma, \kappa}\circ \Gamma_{\tau_1}^{1/2}\right)^{\otimes M}$.
This implies that the singular values of the map $\Gamma_\tau^{-1/2}\circ\mA_{\gamma, \kappa}^{\otimes M}\circ \Gamma_\tau^{1/2}$ on $M$-qubit systems are obtained by multiplying the singular values of the single-qubit map $\Gamma_{\tau_1}^{-1/2}\circ\mA_{\gamma, \kappa}\circ \Gamma_{\tau_1}^{1/2}$.
In addition, since the largest singular value of $\Gamma_{\tau_1}^{-1/2}\circ\mA_{\gamma, \kappa}\circ \Gamma_{\tau_1}^{1/2}$ is 1, we get $s_2(\mA_{\gamma,\kappa}^{\otimes M})=s_2(\mA_{\gamma,\kappa})$ allowing us to only compute the second largest singular value of the single-qubit map $\Gamma_{\tau_1}^{-1/2}\circ\mA_{\gamma, \kappa}\circ \Gamma_{\tau_1}^{1/2}$.
This can be obtained as $s_2(\mA_{\gamma,\kappa}^{\otimes M})=s_2(\mA_{\gamma,\kappa})=\sqrt{1-\gamma}$ by a simple computation. 

Plugging these quantities into \eqref{eq:sample lower bound non-unital} results in
\bal
 N\geq \frac{(1-2\varepsilon)^2\left[\tilde\kappa\left(1-\gamma/2 - \sqrt{1-\gamma + \gamma^2(\kappa - 1/2)^2}\right)\right]^M}{8\ln 2 \left(1-\gamma\right)^{L} },\quad \tilde\kappa\coloneqq\min\{\kappa,1-\kappa\},
 \label{eq:sample lower bound generalized amplitude damping}
\eal
which particularly grows exponentially with the layer number $L$, noting that $0<\gamma<1$.

\section{Relation to previous works on noisy layered circuits}\label{app:previous works}

Theorems~\ref{thm:layer}~and~\ref{thm:layer standard deviation} imply that if an $n$-qubit quantum circuit with super-logarithmic depth suffers from local depolarizing noise (and other types of noise as discussed in Appendix~\ref{app:layered general noise}), super-polynomial number of samples are required, making the computation infeasible even with the help of quantum error mitigation.
The observation that computation quickly becomes useless under the noise effects dates back to the seminal work by Aharonov et al.~\cite{Aharonov1996limitations} (see, e.g., Refs.~\cite{Unruh1995maintaining,Palma1996quantum} for the works that put forward and investigated related ideas). 
Here, we clarify differences between our results and the previous works, in particular Ref.~\cite{Aharonov1996limitations}.

The setting discussed in Ref.~\cite{Aharonov1996limitations} involves a quantum circuit with $d$ layers of unitary gates, where each layer is followed by the action of local depolarizing noise. 
After the $d$\,th layer, each qubit is measured with the computational basis. 
The authors showed that when the depth is super-logarithmic with the qubit number, the probability of obtaining 0 (and 1) cannot be away from 1/2 by a constant amount.
This implies that to realize the ``useful'' computation --- which should not be simulable by a random guess --- the noisy quantum circuit should be restricted to the one with logarithmic depth. 
The authors showed this result by employing the fact --- which was also shown in the same manuscript --- that the relative entropy $S(\rho\|\mbI/2^n)=n-S(\rho)$ decreases exponentially with circuit depth. 
They combined this with the fact that the computational-basis measurement (followed by partial trace) does not increase the relative entropy, showing that the entropy of the probability distribution of the measurement outcome of the first qubit can only be deviated from the maximal value by the amount scaling as $n2^{-d}$.
Therefore, to ensure the probability of obtaining 0 or 1 to be away from 1/2 by a constant amount, one needs to have $d=\mO(\log n)$ to allow the entropy to be away from the maximal value with a large $n$ limit.

Although the above consideration contains important insights, and our results --- as well as many other follow-up works after this such as Refs.~\cite{Kastoryano2013quantum,DePalma2023limitations,Franca2021limitations,Wang2021noise-induced,Wang2021can} --- benefit from their inspirations, the argument there is not directly applicable to the setting involving quantum error mitigation. 
As explained in the main text, the idea of quantum error mitigation is to use quantum and classical resources in hybrid, in which postprocessing computation after the noisy circuit is essential. 
In our framework, this part is included in the trailing process ($\mP_A$ in Fig.~\ref{fig:framework} in the main text) --- indeed, our model allows an arbitrary quantum operation as the trailing process, which also includes classical postprocessing computation.
Importantly, \emph{the trailing process is not unital in general}, and therefore, entropy can decrease in the postprocessing step, preventing the results in Ref.~\cite{Aharonov1996limitations} from being directly carried over.

We solve this problem by first reducing the performance of error mitigation to the distinguishability of two noisy quantum states, which eventually leads to Theorems~\ref{thm:sampling bound}~and~\ref{thm:sample bound standard deviation} providing the necessary sampling cost to achieve the target error mitigation performance with respect to entropic measures.
It is a priori not obvious how entropic quantity is quantitatively related to the operational performance quantifiers such as accuracy--probability and bias--standard deviation of quantum error mitigation. 
Theorems~\ref{thm:sampling bound}~and~\ref{thm:sample bound standard deviation} establish the connections between these quantities, which then allow us to focus on analyzing the entropic measures, where we can apply the previous findings, such as the exponential decay of relative entropy from the maximally mixed state under local depolarizing noise, to show the exponential sampling overhead for the general error mitigation. 

In addition, our quantitative bounds --- which give nearly tight estimates for certain cases --- provide concrete sampling lower bounds beyond just the exponential scaling behavior for layered circuits, which could also be used as a benchmark for error mitigation strategies.   

Let us also remark on how our results are related to the prior works that investigated the capability of general error mitigation, which involves postprocessing operations.
Refs.~\cite{Takagi2021fundamental,Wang2021can} studied the maximum estimator spread, i.e., the range of outcomes of the estimator, imposed on general error mitigation protocols. 
These works showed that, for error mitigation with linear postprocessing, the maximum estimator spread must grow exponentially with the circuit depth for noisy layered circuits under local depolarizing noise. 
Although these results hinted that the exponential samples would be required in general error mitigation, they were not conclusive because the maximum estimator spread provides only a \emph{sufficient} number of samples to ensure a certain accuracy.
Our results close this gap by showing that exponential samples are \emph{necessary}.
Ref.~\cite{DePalma2023limitations} studied the general framework introduced in Ref.~\cite{Takagi2021fundamental} and obtained a concentration bound, which places an upper bound for the probability of getting an estimate away from the value for the maximally mixed state by a certain amount. Their bound implies that, when the circuit depth is $\Omega(\log N)$, the probability of achieving a certain constant accuracy will become exponentially small with respect to the accuracy multiplied by the Lipschitz constant of the estimator. 
However, as the authors pointed out, the Lipschitz constant of the estimator becomes exponentially large for many error mitigation protocols, which severely restricts the applicability of their bound. Also, their bound does not directly provide a sampling lower bound to achieve a certain accuracy. 
Our results, on the other hand, encompass the standard estimators with exponentially large Lipschitz constants and provide the first explicit sampling lower bounds that apply to general error mitigation protocols.

\end{document}